\documentclass[a4paper,UKenglish,cleveref, autoref, thm-restate]{lipics-v2021}
\usepackage{extarrows}
\usepackage{graphicx}
\usepackage{mathtools}
\usepackage[ruled]{algorithm2e}

\allowdisplaybreaks
\nolinenumbers

\theoremstyle{definition}
\newtheorem{defn}[theorem]{Definition}
\theoremstyle{definition}
\newtheorem{notation}[theorem]{Notation}

\newcommand{\Z}{\mathbb{Z}}
\newcommand{\N}{\mathbb{N}}
\newcommand{\Q}{\mathbb{Q}}
\newcommand{\R}{\mathbb{R}}
\newcommand{\C}{\mathbb{C}}
\newcommand{\Rp}{\mathbb{R}_{\geq 0}}
\newcommand{\Qp}{\mathbb{Q}_{\geq 0}}
\newcommand{\Zp}{\mathbb{Z}_{\geq 0}}
\newcommand{\UT}{\mathsf{UT}}
\newcommand{\SL}{\operatorname{SL}}
\newcommand{\GL}{\operatorname{GL}}
\newcommand{\U}{\mathsf{U}}
\newcommand{\Heis}{\operatorname{H}_3}
\newcommand{\Sym}{\operatorname{S}}

\newcommand{\card}{\operatorname{card}}
\newcommand{\supp}{\operatorname{supp}}
\newcommand{\mG}{\mathcal{G}}
\newcommand{\mE}{\mathcal{E}}
\newcommand{\mC}{\mathcal{C}}

\newcommand{\mI}{\mathcal{I}}

\newcommand{\bv}{\boldsymbol{v}}
\newcommand{\bl}{\boldsymbol{\ell}}
\newcommand{\bzer}{\boldsymbol{0}}

\title{On the Identity Problem for Unitriangular Matrices of Dimension Four}
\author{Ruiwen Dong}{Department of Computer Science, University of Oxford}{ruiwen.dong@kellogg.ox.ac.uk}{}{}

\authorrunning{R. Dong}

\Copyright{Ruiwen Dong} 



\category{} 

\relatedversion{} 

\supplement{}

\acknowledgements{}

\EventEditors{Stefan Szeider, Robert Ganian, and Alexandra Silva}
\EventNoEds{3}
\EventLongTitle{47th International Symposium on Mathematical Foundations of Computer Science (MFCS 2022)}
\EventShortTitle{MFCS 2022}
\EventAcronym{MFCS}
\EventYear{2022}
\EventDate{August 22--26, 2022}
\EventLocation{Vienna, Austria}
\EventLogo{}
\SeriesVolume{241}
\ArticleNo{9}

\begin{document}

\maketitle
\begin{abstract}
We show that the Identity Problem is decidable in polynomial time for finitely generated sub-semigroups of the group $\UT(4, \Z)$ of $4 \times 4$ unitriangular integer matrices.
As a byproduct of our proof, we also show the polynomial-time decidability of several subset reachability problems in $\UT(4, \Z)$.
\end{abstract}

\section{Introduction}
Among the most prominent algorithmic problems for matrix semigroups are the \emph{Identity Problem} and the \emph{Membership Problem}.
For the Membership Problem, the input is a finite set of square matrices $A_1, \ldots, A_k$ and a target matrix $A$.
The problem is to decide whether $A$ lies in the semigroup generated by $A_1, \ldots, A_k$.
The Identity Problem is the Membership Problem restricted to the case where $A$ is the identity matrix.
These two problems are closely related to each other, and, as shown in many circumstances, studying the Identity Problem is usually the first step in studying the Membership Problem.

For general matrices, the Membership Problem is undecidable by a classical result of Markov \cite{markov1947certain}.
Indeed, it is one of the earliest undecidability results on algorithmic problems in matrix semigroups.
Most variants of the problem remain undecidable in low dimension. For example, the \emph{Mortality Problem}, which is the Membership Problem in which the target matrix is 0, is undecidable in dimension three~\cite{paterson1970unsolvability}.  
In dimension four, the Membership Problem is undecidable for matrices in $\SL(4, \Z)$ (see~\cite{mikhailova1966occurrence}), while the Identity Problem is undecidable for the set of $4 \times 4$ integer matrices $\mathcal{M}_{4\times 4}(\mathbb{Z})$ (see~\cite{bell2010undecidability}).


However, there has also been steady progress on the decidability side.
The Membership Problem is shown to be decidable for $\GL(2, \Z)$ in \cite{choffrut2005some}.
This decidability result is then extended to $2 \times 2$ integer matrices with nonzero determinant \cite{potapov2017decidability}, and to $2 \times 2$ integer matrices with determinants equal to 0 and $\pm 1$ \cite{potapov2017membership}.
It remains an intricate open problem whether the Membership Problem or the Identity Problem is decidable for $\SL(3, \Z)$.

Recently, there has been more progress on closing the decidability gap by restricting consideration to the class of unitriangular matrices.
It has long been known that the \emph{Group Membership Problem} is decidable for $\UT(n, \Z)$, the group of unitriangular integer matrices of dimension $n$.
The Group Membership Problem asks to decide whether a matrix $A$ lies in the \emph{group} generated by given matrices $A_1, \ldots, A_k$.
In fact, it is decidable for all finitely generated solvable matrix groups \cite{kopytov1968solvability}.
Later, Babai et al. \cite{babai1996multiplicative} showed that the Group Membership Problem for \emph{commuting matrices} can be computed in polynomial time (note that commuting matrices are simultaneously upper-triangularizable).
However, there are significant differences between the group case and the semigroup case.
In fact, for large enough $n$, the \emph{Knapsack Problem} for $\UT(n, \Z)$ is undecidable \cite{konig2016knapsack}.
Given matrices $A_1, \ldots, A_k$ and $A$, the Knapsack Problem asks to decide whether there exist natural numbers $e_1, \ldots, e_k$ such that $A_1^{e_1} \cdots A_k^{e_k} = A$.
From the undecidability of the Knapsack Problem, one can deduce the undecidability of the semigroup Membership Problem for $\UT(n, \Z)$ for large enough $n$ \cite{lefaucheux2022private}.

Nevertheless, there have been some positive decidability results.
The Identity Problem has been shown to be decidable for the group of $3 \times 3$ unitriangular integer matrices $\UT(3, \Z)$ and the Heisenberg groups $\operatorname{H}_{2n+1}$ in \cite{ko2017identity}.
Shortly after, the decidability result was extended to the Membership Problem \cite{colcombet2019reachability}.
Ko et al. left open the problem whether the Identity Problem in $\UT(n, \Z)$ is decidable for $n \geq 4$, as well as finding the smallest $n$ for which the Membership Problem for $\UT(n, \Z)$ becomes undecidable.

The main result of this paper is that the Identity Problem is decidable in polynomial time for $\UT(4, \Z)$.
This further narrows the gap between decidability and undecidability and can be regarded as a first step towards the Membership Problem for $\UT(4, \Z)$.
The foundation of our method is the arguments developed in \cite{colcombet2019reachability} for the Membership Problem of $\UT(3, \Z)$.
However, in order to pass from dimension three to four, we need to introduce additional methods from convex geometry, linear programming and even use the aid of computational algebraic geometry software.
The proof for $\UT(3, \Z)$ heavily relies on the fact that the subgroup generated by commutators of matrices from a given subset of $\{A_1, \ldots, A_k\} \subset \UT(3, \Z)$ is isomorphic to a subgroup of $\Z$.
This is no longer the case for $\UT(4, \Z)$.
However, $\UT(4, \Z)$ is still metabelian \cite{rotman2012introduction}, and its derived subgroup is isomorphic to $\Z^3$. 
Given a finite set $\mG \subseteq \UT(4, \Z)$, we construct elements in $\langle \mG \rangle$ that fall inside the derived subgroup of $\UT(4, \Z)$.
These elements then generate a cone in $\Z^3$ under the isomorphism between the derived subgroup and $\Z^3$.
The possible shapes of this cone will determine the Identity Problem.

There is strong evidence that the new techniques introduced in this paper can help tackle the Identity Problem for $\UT(n, \Z)$ with $n \geq 5$.

\section{Preliminaries}
Denote by $\UT(4, \Z)$ the group of upper triangular integer matrices with ones on the diagonal:
\[
\UT(4, \Z) \coloneqq \left\{
    \begin{pmatrix}
    1 & a & d & f \\
    0 & 1 & b & e  \\
    0 & 0 & 1 & c \\
    0 & 0 & 0 & 1 \\
    \end{pmatrix}
    \middle| a, b, c, d, e, f \in \Z
    \right\}.
\]
Denote its normal subgroups
\[
\U_1 \coloneqq \left\{
    \begin{pmatrix}
    1 & 0 & d & f \\
    0 & 1 & 0 & e  \\
    0 & 0 & 1 & 0 \\
    0 & 0 & 0 & 1 \\
    \end{pmatrix}
    \middle| d, e, f \in \Z
    \right\}, \quad
\U_2 \coloneqq \left\{
    \begin{pmatrix}
    1 & 0 & 0 & f \\
    0 & 1 & 0 & 0  \\
    0 & 0 & 1 & 0 \\
    0 & 0 & 0 & 1 \\
    \end{pmatrix}
    \middle| f \in \Z
    \right\}
\]
in the lower central series:
$
\UT(4, \Z) \trianglerighteq \U_1 = [\UT(4, \Z), \UT(4, \Z)] \trianglerighteq \U_2 = [\UT(4, \Z), \U_1]
$
(see \cite[Chapter~5]{rotman2012introduction}).
In particular, $\U_1$ and $\U_2$ are respectively the derived subgroup and the centre of $\UT(4, \Z)$.
For convenience, we introduce the following notations:
\[
UT(a,b,c; d,e,f) \coloneqq \begin{pmatrix}
    1 & a & d & f \\
    0 & 1 & b & e  \\
    0 & 0 & 1 & c \\
    0 & 0 & 0 & 1 \\
    \end{pmatrix}, \quad
U_1(d,e,f) \coloneqq UT(0,0,0; d,e,f).
\]
There are surjective group homomorphisms $\varphi_0\colon \UT(4, \Z) \rightarrow \Z^3$ defined by 
\[
\varphi_0(UT(a,b,c; d,e,f)) = (a,b,c),
\]
with $\ker(\varphi_0) = \U_1$, and $\varphi_1\colon \U_1 \rightarrow \Z^2$,
\[
\varphi_1(U_1(d,e,f)) = (d,e),
\]
with $\ker(\varphi_1) = \U_2$.
Moreover, $\U_1$ is itself abelian, with a natural isomorphism $\tau \colon \U_1 \xrightarrow{\sim} \Z^3$:
\[
\tau(U_1(d,e,f)) = (d,e,f).
\]
Denote by $\tau_d$ the projection $U_1(d,e,f) \mapsto d$, $\tau_e$ the projection $U_1(d,e,f) \mapsto e$,
and $\tau_{f}$ the projection $U_1(d,e,f) \mapsto f$.
Then, $\varphi_1 = (\tau_d, \tau_e)$ and $\tau = (\tau_d, \tau_e, \tau_f)$.

Finally, define the subgroup of $\UT(4, \Z)$:
\[
\U_{10} \coloneqq \left\{
    U_1(0,e,f)
    \mid e, f \in \Z
    \right\} \trianglelefteq \U_1.
\]
For a finite set of matrices $\mG = \{A_1, \ldots, A_k\}$, denote by $\langle \mG \rangle$ the semigroup generated by $\mG$.
In this paper, we are concerned with the following problems.
\begin{defn}
Let $G$ be a monoid of matrices, and $H$ a subset of $G$.
\begin{enumerate}[(i)]
    \item The \emph{Identity Problem} in $G$ asks, given a finite set of matrices $\mG$ in $G$, whether $I \in \langle \mG \rangle$.
    If this is the case, we say that the identity matrix is \emph{reachable}.
    \item The \emph{$H$-Reachability Problem} in $G$ asks, given a finite set of matrices $\mG$ in $G$, whether $H \cap \langle \mG \rangle \neq \emptyset$.
    If this is the case, we say that $H$ is \emph{reachable}.
\end{enumerate}
\end{defn}
The main result of this paper is that the Identity Problem in $\UT(4, \Z)$ is decidable in polynomial time, with respect to the number of bits required to encode all the entries of the matrices in $\mG$ (each matrix $UT(a,b,c,d,e,f)$ is encoded by the entries $a, b, c, d, e, f$).

It turns out that the three problems: Identity Problem, $\U_2$-Reachability and $\U_{10}$-Reachability are interconnected and it is more convenient to devise algorithms that decide them simultaneously.
A trivial observation is that, because $I \in \U_2 \subset \U_{10}$, a positive instance of the Identity Problem is also a positive instance of $\U_2$-Reachability;
and a positive instance of $\U_2$-Reachability is also a positive instance of $\U_{10}$-Reachability.

The following definitions will be used throughout this paper.
\begin{defn}[String, product and Parikh vector]
Let $\mG = \{A_1, \ldots, A_k\}$ be a fixed set of matrices in $\UT(4, \Z)$.
A \emph{string} of $\mG$ is an expression $B_1 B_2 \cdots B_m$ such that $B_i \in \mG, i = 1, \ldots, m$.
The \emph{product} of a string $B_1 B_2 \cdots B_m$ is the matrix $P \in \UT(4, \Z)$ such that $P = B_1 B_2 \cdots B_m$.
The \emph{Parikh vector} of a string $B_1 B_2 \cdots B_m$ is the vector $\bl = (\ell_1, \ldots, \ell_k) \in \Zp$ where
\[
\ell_j = \card(\{i \mid B_i = A_j\}), j = 1, \ldots, k.
\]
When $\mG$ is clear from the context, we simply use the term ``string'' instead of ``string of $\mG$''.
\end{defn}

For an integer $n \geq 1$, the \emph{Heisenberg group} of dimension $2n+1$ is the group $\operatorname{H}_{2n+1}$ of $(n+2) \times (n+2)$ integer matrices of the form
$
H = \begin{pmatrix}
        1 & \boldsymbol{a} & c \\
        0 & I_{n} & \boldsymbol{b}^{\top} \\
        0 & 0 & 1 \\
\end{pmatrix},
$
where $\boldsymbol{a}, \boldsymbol{b} \in \Z^{n}$, $c \in \Z$.
The following result comes from \cite{ko2017identity} and \cite{colcombet2019reachability}.
\begin{lemma}[{\cite[Theorem~7]{colcombet2019reachability}}]\label{lem:Heis}
The Identity Problem and the Membership Problem in $\operatorname{H}_{2n+1}$ are decidable for all $n \geq 1$.
\end{lemma}

\section{Identity problem, $\U_2$- and $\U_{10}$-Reachability in $\UT(4, \Z)$}
In this section, we construct algorithms that decide the Identity Problem, $\U_2$-Reachability and $\U_{10}$-Reachability in $\UT(4, \Z)$.

\subsection{Overview of decision strategy}

For any set of vectors $\bv_1, \ldots, \bv_l \in \R^n$, denote by
\[
\langle \bv_1, \ldots, \bv_l \rangle_{\Rp} \coloneqq \left\{\sum_{i = 1}^{l} r_i \bv_i \middle| r_i \in \Rp, i = 1, \ldots, l\right\}
\]
the $\Rp$-cone generated by $\bv_1, \ldots, \bv_l$, and by
$\langle \bv_1, \ldots, \bv_l \rangle_{\R}$
the $\R$-vector space spanned by $\bv_1, \ldots, \bv_l$.

Let $\mG = \{A_1, \ldots, A_k\}$ be a set of matrices in $\UT(4, \Z)$, for which we want to decide the Identity Problem, $\U_2$-Reachability and $\U_{10}$-Reachability.
Define the $\Rp$-cone
\begin{equation}\label{eq:defC}
\mC \coloneqq \langle \varphi_0(A_1), \ldots, \varphi_0(A_k) \rangle_{\Rp},
\end{equation}
and denote by $\mC^{lin}$ its lineality space, i.e.\ the largest linear subspace (by inclusion) contained in $\mC$.
In particular, $\mC^{lin} = \mC \cap - \mC$.
A basis of $\mC^{lin}$ can be effectively computed in polynomial time \cite{schrijver1998theory}.
For any matrix $A_i \in \mG$, the projection $\varphi_0(A_i)$ can be either in $\mC^{lin}$ or in $\mC \setminus \mC^{lin}$.
However, in order to reach $\U_1$, which contains the identity matrix, $\U_2$ and $\U_{10}$, one can only use matrices $A_i$ with $\varphi_0(A_i) \in \mC^{lin}$.
This is formally stated by the following proposition. 
\begin{restatable}{prop}{Clin}\label{prop:Clin}
If the product of a string $B_1 \cdots B_m$ is in $\U_1$, then every $B_j, j = 1, \ldots m,$ must be in the set $\{A_i \in \mG \mid \varphi_0(A_i) \in \mC^{lin}\}$.
\end{restatable}

\begin{proof}
Suppose on the contrary that some $B_j$ satisfies $\varphi_0(B_j) \in \mC \setminus \mC^{lin}$.

Since $\varphi_0$ is a group homomorphism, we have
\[
B_1 \cdots B_m \in \U_1 \iff \varphi_0(B_1 \cdots B_m) = \bzer \iff \sum_{i=1}^m \varphi_0(B_i) = \bzer.
\]
Therefore, $-\varphi_0(B_j) = \sum_{i \neq j} \varphi_0(B_i) \in \mC$.

Hence, the linear subspace $\varphi_0(B_j) \R = \langle \varphi_0(B_j), -\varphi_0(B_j) \rangle_{\Rp}$ is contained in $\mC$.
This yields $\varphi_0(B_j) \R \subseteq \mC^{lin}$, a contradiction to $\varphi_0(B_j) \in \mC \setminus \mC^{lin}$.
\end{proof}

The overall strategy for constructing our algorithm is to use induction on $\card(\mG)$.
If $\card(\mG) = 0$, then the answers to the Identity Problem, $\U_2$-Reachability and $\U_{10}$-Reachability are all negative.
Suppose now that we have an algorithm that decides all three problems for every set of at most $k-1$ matrices,
we will construct an algorithm that decides them for a set of $k$ matrices $\mG = \{A_1, \ldots, A_k\}$.
By Proposition \ref{prop:Clin}, if some matrix $A_i$ satisfies $\varphi_0(A_i) \in \mC \setminus \mC^{lin}$, then we can discard it without changing the answer to the Identity Problem or $\U_2, \U_{10}$-Reachability.
This decreases the number of elements in $\mG$, and an algorithm is available by the induction hypothesis on $\card(\mG)$.
Hence, we can suppose that every $A_i \in \mG$ satisfies $\varphi_0(A_i) \in \mC^{lin}$,
so $\mC = \mC^{lin}$ is a linear space.

Since $\varphi_0(A_i) \in \Z^3$, $\mC$ is a linear subspace of $\R^3$.
We identify cases according to the dimension of $\mC$,
with each of the following four subsections treating the case of dimension 3, 1, 0, 2.
The pseudocode of the decision procedure for the Identity Problem is given here as a reference point for the detailed case analysis. 
The decision procedures for $\U_{2}$-reachability and $\U_{10}$-reachability follow similar patterns and their pseudocode is given in Appendix~\ref{app:alg}.
Note that the decision procedure for the Identity Problem invokes the decision procedure for $\U_{2}$-reachability as a subroutine.
Similarly, the decision procedure for $\U_{2}$-reachability will invoke the decision procedure for $\U_{10}$-reachability as a subroutine.

\begin{algorithm}[H]
\caption{IdentityProblem(): deciding the Identity Problem for a subset of $\UT(4, \Z)$.}
\label{alg:IP}
\begin{description}
\item[Input:] 
A set $\mathcal{G} = \{A_1, \ldots, A_k\}$ of matrices in $\UT(4, \Z)$.
\item[Output:] True or False.
\end{description}
\begin{enumerate}[Step 1:]
    \item Compute the cone $\mathcal{C}$ and its lineality space $\mathcal{C}^{lin}$.
    For $i = 1, \ldots, k$, if some \\
    $\varphi_0(A_i)$ is not in $\mathcal{C}^{lin}$, return IdentityProblem($\mG \setminus \{A_i\}$).
    \item
    \begin{enumerate}
        \item If $\dim(\mathcal{C}) = 3$, return True.
        \item If $\dim(\mathcal{C}) = 1$, return True if the condition in Proposition~\ref{prop:case1id}(i) is satisfied, otherwise return False.
        \item If $\dim(\mathcal{C}) = 0$, return True if $\tau(A_i), i = 1, \ldots, m$ generate a semigroup \\
        containing $\bzer$, otherwise return False.
        \item If $\dim(\mathcal{C}) = 2$, compute a non-zero vector $(p,q,r) \in \Q^3$ orthogonal to $\mathcal{C}$.
    \begin{enumerate}
        \item If $p = 0$, but $q, r$ are not zero, or $r = 0$, but $q, p$ are not zero. \\
        Compute $L_0$, if $\supp(L_0) = \{1, \ldots, k\}$, return True, otherwise return IdentityProblem($\{A_i \mid i \in \supp(L_0)\}$).
        \item If $p = r = 0$, problem reduces to Identity Problem in $\operatorname{H}_5$.
        \item If $p = q = 0, r \neq 0$ or $r = q = 0, p \neq 0$, compute $A'_i$ as in \eqref{eq:defAprime}.
        Return U2Reachability($A_1', \ldots, A_k'$) (see Appendix~\ref{app:alg}).
    \end{enumerate}
    \end{enumerate}
\end{enumerate}
\end{algorithm}

We now give an overview of the motivation behind classifying cases according to the dimension of $\mC$.
As a convention, we always use $A_i, i = 1, \ldots, k$ to denote elements of the fixed generating set $\mG$, and Greek letters to denote their entries, i.e.\ $A_i = UT(\alpha_i, \beta_i, \kappa_i; \delta_i, \epsilon_i, \phi_i)$.
We use $B_i, i = 1, \ldots, m$ to denote arbitrary elements in $\langle \mG \rangle$ (when appearing in \emph{strings}, they are elements in $\mG$), and Latin letters to denote their entries, i.e.\ $B_i = UT(a_i, b_i, c_i, d_i, e_i, f_i)$.
The variables $B_i$ can depend on the context.

First of all, we need some results on the structure of products in $\UT(4, \Z)$.
For a positive integer $m$, denote by $\Sym_m$ the permutation group of the set $\{1, \ldots, m\}$.
Throughout this paper, given some matrices $B_1, \ldots, B_m \in \langle \mG \rangle$, we will often be computing the product of strings of the form $B_{\sigma(1)}^t \cdots B_{\sigma(m)}^t$, where $\sigma \in \Sym_m$ and $t \in \Zp$.
The overall idea is to find various strings $B_{\sigma(1)}^t \cdots B_{\sigma(m)}^t$ whose product is in $\U_1 \overset{\tau}{\cong} \Z^3$, then use them to generate an abelian semigroup containing the identity matrix.
Let us define the following important values and abbreviations that will be used throughout this paper.
These complicated formulas are related to the \emph{logarithm} of the matrices $B_i$, and readers can for the time being ignore their exact form and treat them as black boxes.
\begin{notation}\label{not:latin}
Given a series of matrices $B_1, \ldots, B_m$ where $B_i = UT(a_i, b_i, c_i; d_i, e_i, f_i), i = 1, \ldots, m$, we introduce the following notation:
\begin{enumerate}[(i)]
    \item For $\sigma \in \Sym_m, t \in \Zp$, 
    \begin{equation}
    B(\sigma, t) \coloneqq B_{\sigma(1)}^t \cdots B_{\sigma(m)}^t.
    \end{equation}
    \item For $\sigma \in \Sym_m$,
    \begin{multline}\label{eq:defXsigma}
     D_{\sigma} \coloneqq \sum_{i < j} a_{\sigma(i)} b_{\sigma(j)} + \frac{1}{2}\sum_{i=1}^{m}a_{i} b_{i}, \quad 
    E_{\sigma} \coloneqq \sum_{i < j} b_{\sigma(i)} c_{\sigma(j)} + \frac{1}{2}\sum_{i=1}^{m}b_{i} c_{i}, \\
     F_{\sigma} \coloneqq \sum_{i < j < k} a_{\sigma(i)} b_{\sigma(j)} c_{\sigma(k)} + \frac{1}{2}\sum_{i < j}(a_{\sigma(i)} b_{\sigma(i)} c_{\sigma(j)} + a_{\sigma(i)} b_{\sigma(j)} c_{\sigma(j)}) + \frac{1}{6}\sum_{i = 1}^m a_{i} b_{i} c_{i}, \\
     \\
        G_{\sigma} \coloneqq \sum_{i < j} (a_{\sigma(i)} e_{\sigma(j)} + d_{\sigma(i)} c_{\sigma(j)} - \frac{1}{2}a_{\sigma(i)} b_{\sigma(j)} c_{\sigma(j)} - \frac{1}{2} a_{\sigma(i)} b_{\sigma(i)} c_{\sigma(j)}) \\
     + \frac{1}{2}\sum_{i=1}^m (a_i e_i + d_i c_i - a_i b_i c_i).
    \end{multline}
    \item For $i = 1, \ldots, m$,
    \begin{equation}\label{eq:defXi}
    D_{i} \coloneqq d_i - \frac{1}{2}a_i b_i, \quad
    E_{i} \coloneqq e_i - \frac{1}{2}b_i c_i, \quad
    F_{i} \coloneqq f_i - \frac{1}{2}(a_i e_i + d_i c_i) + \frac{1}{3}a_i b_i c_i.
    \end{equation}
\end{enumerate}
\end{notation}

The following proposition gives an exact expression for $B(\sigma, t)$.
Because of the heavily computational nature of most of our propositions, their proofs are given in Appendix~\ref{app:proofs}.

\begin{restatable}{prop}{exactpr}\label{prop:exactexpr}
Let $B_i = UT(a_i, b_i, c_i; d_i, e_i, f_i), i = 1, \ldots, m$, $\sigma \in \Sym_m, t \in \Zp$, then
\begin{multline}
  B(\sigma, t) = UT\left(t \sum_{i=1}^m a_i, t \sum_{i=1}^m b_i, t \sum_{i=1}^m c_i; \right.\\ 
  \left. t^2 D_{\sigma} + t \sum_{i=1}^m D_i, t^2 E_{\sigma} + t \sum_{i=1}^m E_i, t^3 F_{\sigma} + t^2 G_{\sigma} + t \sum_{i=1}^m F_i\right).
\end{multline}
\end{restatable}
Notice that $B(\sigma, t) \in \U_1$ if and only if $\sum_{i=1}^m a_i = \sum_{i=1}^m b_i = \sum_{i=1}^m c_i = 0$, a condition that does not depend on the value of $t$.

Proposition \ref{prop:exactexpr} shows that, if $B(\sigma, t)$ is in $\U_1$, then as $t \rightarrow \infty$, the asymptotic behaviour of $\tau(B(\sigma, t))$ approaches the vector $(t^2 D_{\sigma}, t^2 E_{\sigma}, t^3 F_{\sigma})$, provided that $D_{\sigma}, E_{\sigma}, F_{\sigma}$ do not vanish.
Therefore, the hope is that, as $t, \sigma$ vary, the vectors $(t^2 D_{\sigma}, t^2 E_{\sigma}, t^3 F_{\sigma})$ can generate $\R^3$ as an $\Rp$-cone, barring a few degenerate cases.
If these degenerate cases do not happen, then the different vectors $\tau(B(\sigma, t))$ will also generate $\R^3$ as an $\Rp$-cone.
In particular, the identity element in $\R^3$ can be generated by $\tau(B(\sigma, t))$ as an additive semigroup, giving a positive answer to the Identity Problem.
For the degenerate cases, they will be treated individually.
As it will turn out, there are only two types of degeneracy (which may occur simultaneously):
\begin{enumerate}[(i)]
    \item $F_{\sigma} = 0$ for all $\sigma$.
    \item For some $p, r \in \Q$, possibly zero, we have $pD_{\sigma} = rE_{\sigma}$ for all $\sigma$.
\end{enumerate}
When (i) occurs, the asymptotic behaviour of $\tau(B(\sigma, t))$ approaches the vector \\
$(t^2 D_{\sigma}, t^2 E_{\sigma}, t^2 G_{\sigma})$, since $G_{\sigma}$ is the second most dominant term after $F_{\sigma}$. 
This situation reminds us of the Identity Problem for $\Heis$, and can be solved in a similar way.
When (ii) occurs, the vectors $(t^2 D_{\sigma}, t^2 E_{\sigma}, t^3 F_{\sigma})$ are constrained to a strict linear subspace of $\R^3$. Hence, in order to describe the $\Rp$-cone generated by the vectors $\tau(B(\sigma, t))$, one needs to consider the sub-dominant terms as well, i.e. the terms $t \sum_{i=1}^m D_i, t \sum_{i=1}^m E_i$.

The rest of this paper aims to formalize this idea.
We first exhibit a series of lemmas that characterise these degenerate cases.
Our first lemma shows that, supposing $B(\sigma, t) \in \U_1$, then degenerate case (ii) happens if and only if $\langle \varphi_0(B_1), \ldots, \varphi_0(B_m) \rangle_{\R}$ is degenerate (i.e.\ of dimension at most 2).

\begin{restatable}{lem}{DERtwo}\label{lem:DER2}
Given $p, r \in \R$ and $m \geq 2$. 
Suppose $\sum_{i=1}^m a_i = \sum_{i=1}^m b_i = \sum_{i=1}^m c_i = 0$.
The two following statements are equivalent:
\begin{enumerate}[(i)]
    \item For all $\sigma \in \Sym_m$, $p D_{\sigma} = r E_{\sigma}$.
    \item Either $b_i = 0$ for all $i = 1, \ldots, m$, or there exist $q \in \R$, such that $p a_i + q b_i + r c_i = 0$ for all $i = 1, \ldots, m$.
\end{enumerate}
\end{restatable}

The next lemma shows that if $B(\sigma, t) \in \U_1$, then by ``inverting'' $\sigma$, we get a permutation $\sigma'$ such that $(D_{\sigma}, E_{\sigma})$ and $(D_{\sigma'}, E_{\sigma'})$ are opposites of one another.
\begin{restatable}{lem}{DEinv}\label{lem:DEinv}
Suppose $\sum_{i=1}^m a_i = \sum_{i=1}^m b_i = \sum_{i=1}^m c_i = 0$, $m \geq 2$.
For every $\sigma \in \Sym_m$, there exists $\sigma' \in \Sym_m$, such that $(D_{\sigma'}, E_{\sigma'}) = - (D_{\sigma}, E_{\sigma})$.
\end{restatable}

We then show that, if $B(\sigma, t) \in \U_1$, then the value of $F_{\sigma}$ for different $\sigma \in \Sym_m$ sums up to zero:

\begin{restatable}{lem}{Fsumzero}\label{lem:Fsum0}
Suppose $\sum_{i=1}^m a_i = \sum_{i=1}^m b_i = \sum_{i=1}^m c_i = 0$, where $m \geq 3$.
Then we have $\sum_{\sigma \in \Sym_m} F_{\sigma} = 0$.
\end{restatable}

The last lemma characterizes situations where the aforementioned degenerate case (i) happens.
Its proof relies on the aid of a computational algebraic geometry software due to the complexity of the expressions $F_{\sigma}$.

\begin{restatable}{lem}{FRone}\label{lem:FR1}
Let $m = 4$.
Suppose $\sum_{i=1}^4 a_i = \sum_{i=1}^4 b_i = \sum_{i=1}^4 c_i = 0$.
Then, $F_{\sigma} = 0$ for all $\sigma \in \Sym_4$, if and only if at least one of the following four conditions holds:
\begin{enumerate}[(i)]
    \item $a_1 = a_2 = a_3 = a_4 = 0$.
    \item $b_1 = b_2 = b_3 = b_4 = 0$.
    \item $c_1 = c_2 = c_3 = c_4 = 0$.
    \item 
    $
    \operatorname{rank}
    \begin{pmatrix}
    a_1 & a_2 & a_3 & a_4 \\
    b_1 & b_2 & b_3 & b_4  \\
    c_1 & c_2 & c_3 & c_4 \\
    \end{pmatrix}
    \leq 1.
    $
\end{enumerate}
\end{restatable}

A common idea of Lemma~\ref{lem:DER2} and Lemma~\ref{lem:FR1} is that the degeneracy of $(D_{\sigma}, E_{\sigma}, F_{\sigma})$ is related to the degeneracy of $\varphi_0(B_1), \ldots, \varphi_0(B_m)$.
Hence, it is natural to consider the degeneracy of the vectors $\varphi_0(A_i), i = 1, \ldots, k$, where $A_i \in \mG$ are the elements of the generating set.
This degeneracy is described by the dimension of the linear space $\mC$ discussed at the beginning of the section.
This justifies the classification according to $\dim(\mC)$.
We now begin the case analysis.

\subsection{$\mC$ has dimension 3}\label{subsection:case3}
The main idea of this case is that, for a well chosen set of matrices $B_1, B_2, B_3, B_4 \in \langle \mG \rangle$, the vectors $(D_{\sigma}, E_{\sigma}, F_{\sigma}), \sigma \in \Sym_4$, are not degenerate and the asymptotic behaviour of $\tau(B(\sigma, t))$ approaches the vector $(t^2 D_{\sigma}, t^2 E_{\sigma}, t^3 F_{\sigma})$, leading to a positive answer to the Identity Problem.

Let $B_1, B_2, B_3, B_4 \in \langle \mG \rangle$ with $B_i = UT(a_i, b_i, c_i; d_i, e_i, f_i), i = 1, \ldots, 4$ be such that 
\begin{equation}\label{eq:case3sum4}
    \sum_{i=1}^4 \varphi_0(B_i) = 0
\end{equation}
and
\begin{equation}\label{eq:case3span4}
    \langle \varphi_0(B_1), \varphi_0(B_2), \varphi_0(B_3), \varphi_0(B_4) \rangle_{\Rp} = \mC = \R^3.
\end{equation}
Equation \eqref{eq:case3sum4} shows that $B(\sigma, t) \in \U_1$ for all $\sigma \in \Sym_4, t \in \Zp$.

The following lemma shows that the $d, e$-coordinates of different $\tau(B(\sigma, t))$ generate $\R^2$ as an $\Rp$-cone.
\begin{lemma}\label{lem:case3R2}
Assuming \eqref{eq:case3sum4} and \eqref{eq:case3span4}, we have $\langle \{ \varphi_1(B(\sigma, t)) \mid \sigma \in \Sym_4, t \in \Z \} \rangle_{\Rp} = \R^2$.
\end{lemma}

\begin{proof}
First, we claim that
$
\langle \{ (D_{\sigma}, E_{\sigma}) \mid \sigma \in \Sym_4\} \rangle_{\R} = \R^2.
$

In fact, suppose to the contrary that $\langle \{ (D_{\sigma}, E_{\sigma}) \mid \sigma \in \Sym_4\} \rangle_{\R}$ has dimension at most 1. 
Then there exist $p, r \in \R$, not both zero, such that for all $\sigma \in \Sym_4$, $p D_{\sigma} = r E_{\sigma}$.
By Lemma \ref{lem:DER2}, this means that either $b_i = 0$ for all $i$ or there exists some $q \in \R$ such that $p a_i + q b_i + r c_i = 0$ for all $i$.
In both cases, the $\R$-linear subspace spanned by $\varphi_0(B_1), \varphi_0(B_2), \varphi_0(B_3), \varphi_0(B_4)$ has dimension at most 2,
contradicting Equation \eqref{eq:case3span4}.
This proves the claim.
Hence, there exist $\sigma_1, \sigma_2 \in \Sym_4$ such that $(D_{\sigma_1}, E_{\sigma_1})$ and $(D_{\sigma_2}, E_{\sigma_2})$ span $\R^2$ as an $\R$-linear space.

Next, by Lemma \ref{lem:DEinv}, there exist $\sigma'_1, \sigma'_2 \in \Sym_4$ such that $(D_{\sigma'_1}, E_{\sigma'_1}) = - (D_{\sigma_1}, E_{\sigma_1})$ and $(D_{\sigma'_2}, E_{\sigma'_2}) = - (D_{\sigma_2}, E_{\sigma_2})$.
It follows that $(D_{\sigma_1}, E_{\sigma_1}), (D_{\sigma_2}, E_{\sigma_2}), (D_{\sigma'_1}, E_{\sigma'_1}), (D_{\sigma'_2}, E_{\sigma'_2})$ generate $\R^2$ as an $\Rp$-cone, and all four vectors are non-zero.

Finally, consider the products $B(\sigma, t)$ with $\sigma \in \{\sigma_1, \sigma_2, \sigma'_1, \sigma'_2\}$.
By Proposition \ref{prop:exactexpr}, when $t \rightarrow +\infty$, we have $\varphi_1(B(\sigma, t)) = (D_{\sigma}, E_{\sigma}) t^2 + O(t)$.
Therefore, when $t$ is large enough, the angle between $\varphi_1(B(\sigma, t))$ and $(D_{\sigma}, E_{\sigma})$ tends to zero, for all $\sigma \in \{\sigma_1, \sigma_2, \sigma'_1, \sigma'_2\}$.
Hence, for large enough $t$, $\varphi_1(B(\sigma_1, t)), \varphi_1(B(\sigma_2, t)), \varphi_1(B(\sigma'_1, t)), \varphi_1(B(\sigma'_2, t))$ generate $\R^2$ as an $\Rp$-cone.
This proves the Lemma.
\end{proof}

The next proposition shows that as $\sigma, t$ vary, the vectors $\tau(B(\sigma, t))$ generate $\R^3$ as an $\Rp$-cone.
\begin{proposition}\label{prop:case3R3}
Assuming \eqref{eq:case3sum4} and \eqref{eq:case3span4}, we have $\langle \{ \tau(B(\sigma, t)) \mid \sigma \in \Sym_4, t \in \Z \} \rangle_{\Rp} = \R^3$.
\end{proposition}
\begin{proof}
First, note that all $B(\sigma, t)$ have integer coefficients.
By Lemma \ref{lem:case3R2}, there exist elements $P_1, P_2, P_3 \in \langle \{ B(\sigma, t) \mid \sigma \in \Sym_4, t \in \Z \} \rangle$ such that $\varphi_1(P_i), i = 1, 2, 3$ generate $\R^2$ as an $\Rp$-cone (see Figure \ref{fig:3dGen} for an illustration.).

Next, the idea is to find two additional matrices $P_+, P_- \in \{ B(\sigma, t) \mid \sigma \in \Sym_4, t \in \Z \}$, whose images under $\tau$ are relatively close to the $f$-axis in $\R^3$.
By Lemmas \ref{lem:Fsum0} and \ref{lem:FR1}, there exist $\sigma_+, \sigma_- \in \Sym_4$ such that $F_{\sigma_+} > 0, F_{\sigma_-} < 0$.
Indeed, by condition \eqref{eq:case3span4}, none of the four conditions of Lemma \ref{lem:FR1} hold. Thus there exists $\sigma \in \Sym_4$ such that $F_{\sigma} \neq 0$.
Then Lemma \ref{lem:Fsum0} shows we can find $\sigma_+$ and $\sigma_-$ such that $F_{\sigma_+} > 0$ and $F_{\sigma_-} < 0$.

By Proposition \ref{prop:exactexpr}, when $t \rightarrow +\infty$, we have $\tau_f(B(\sigma_{+}, t)) = F_{\sigma_{+}} t^3 + O(t^2)$ and $\tau_f(B(\sigma_{-}, t)) = F_{\sigma_{-}} t^3 + O(t^2)$,
whereas $\tau_d(B(\sigma_{\pm}, t)) = O(t^2)$ and $\tau_e(B(\sigma_{\pm}, t)) = O(t^2)$.
Therefore, when $t$ is large enough, the angle between $\tau(B(\sigma_{+}, t))$ and $(0, 0, 1)$ tends to zero, as well as the angle between $\tau(B(\sigma_{-}, t))$ and $(0, 0, -1)$.

Finally, we claim that there exists $t$ such that $\tau(P_1), \tau(P_2), \tau(P_3)$, $\tau(B(\sigma_+, t))$, $\tau(B(\sigma_-, t))$ generate $\R^3$ as an $\Rp$-cone.
See Figure \ref{fig:3dGen} for an illustration.
To justify this claim, suppose to the contrary that for every $t$, the $\Rp$-cone spanned by the five vectors $\tau(P_1), \tau(P_2), \tau(P_3)$, $\tau(B(\sigma_+, t))$, $\tau(B(\sigma_-, t))$ is a proper subset of $\R^3$.
In other words, if we denote by $\langle \cdot , \cdot \rangle$ the canonical inner product of $\R^3$, then there exists a vector $\bv_t$ with norm 1, such that $\langle \bv_t , \tau(P_i) \rangle \geq 0, i = 1, 2, 3$ and $\langle \bv_t , \tau(B(\sigma_{\pm}, t)) \rangle \geq 0$. 
For example, we can take $\bv_t$ to be any normalized vector in the dual of the cone generated by these five vectors (\cite[Chapter~2.6]{boyd2004convex}).
By the compactness of the unit sphere, $\{\bv_t\}_{t \in \N}$ has a limit point $\bv$.
We have $\langle \bv , \tau(P_i) \rangle \geq 0, i = 1, 2, 3$, so $\bv$ is not orthogonal to the $f$-axis, otherwise $\tau(P_i), i = 1, 2, 3$ would all be on the same side of a hyperplane passing through the $f$-axis, contradicting the fact that their $d,e$-coordinates generate $\R^2$ as an $\Rp$-cone.
Hence, $\langle \bv , (0,0,1) \rangle \neq 0$.
Without loss of generality, suppose $\langle \bv , (0,0,1) \rangle < 0$.
When $t \rightarrow \infty$, the angle between $(0,0,1)$ and $\tau(B(\sigma_+, t))$ tends to zero. 
Therefore, for all large enough $t$, we have $\langle \bv , \tau(B(\sigma_+, t)) \rangle < 0$.
Since $\bv$ is a limit point of $\{\bv_t\}_{t \in \N}$, there exists a large enough $t$ such that $\langle \bv_t , \tau(B(\sigma_+, t)) \rangle < 0$.
This contradicts the fact that $\langle \bv_t , \tau(B(\sigma_+, t)) \rangle \geq 0$ for all $t$.
\end{proof}

\begin{figure}[h]
    \centering
    \begin{minipage}[t]{.45\textwidth}
        \centering
        \includegraphics[width=1.0\textwidth,height=1.0\textheight,keepaspectratio, trim={4.5cm 0.75cm 4.5cm 0.3cm},clip]{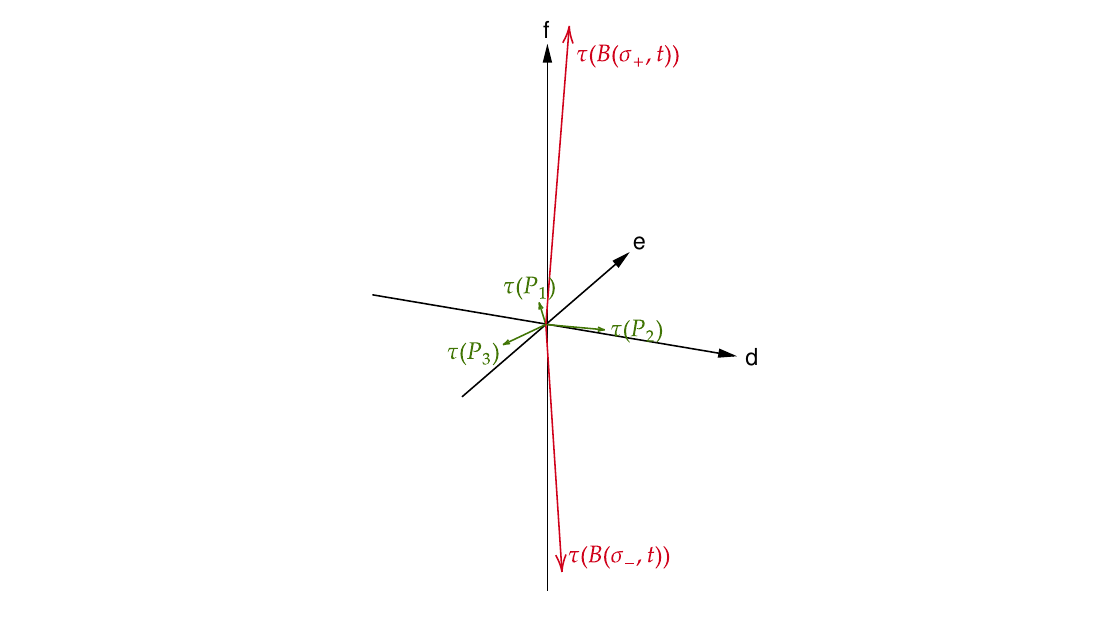}
        \caption{Illustration of the five vectors constructed in Proposition \ref{prop:case3R3}.}
        \label{fig:3dGen}
    \end{minipage}
    \hfill
    \begin{minipage}[t]{0.45\textwidth}
        \centering
        \includegraphics[width=1.0\textwidth,height=1.0\textheight,keepaspectratio, trim={5.5cm 0.75cm 3.5cm 0.3cm},clip]{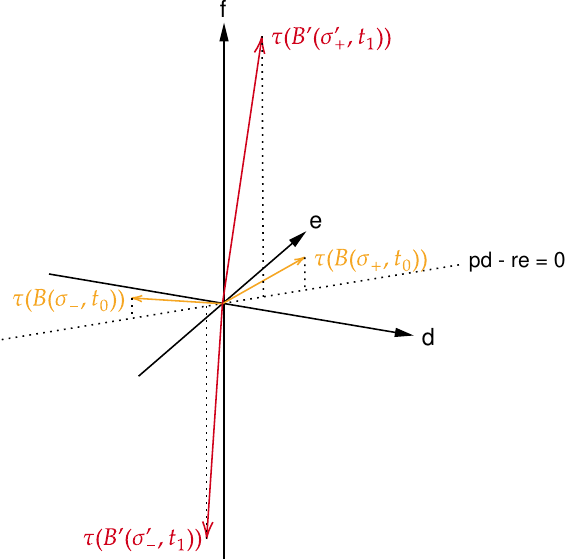}
        \caption{Illustration of the four vectors constructed in Proposition \ref{prop:case21reach2D}.}
        \label{fig:2dLattice}
    \end{minipage}
\end{figure}

\begin{corollary}\label{cor:case3reachI}
When $\mC$ has dimension 3, the identity matrix is reachable (and hence also $\U_2$ and $\U_{10}$).
\end{corollary}
\begin{proof}
By Proposition \ref{prop:case3R3}, one can find $Q_1, Q_2, Q_3, Q_4 \in \langle \mG \rangle \cap \U_1$ such that $\tau(Q_i), i=1, \ldots, 4$ generate $\R^3$ as an $\Rp$-cone.
In particular, $- \tau(Q_1) \in \langle \tau(Q_1), \tau(Q_2), \tau(Q_3), \tau(Q_4) \rangle_{\Rp}$.
So there exist $x_i \in \Rp, i=1, \ldots, 4$, not all zero, such that $\sum_{i=1}^{4} x_i \tau(Q_i) = \bzer$.
Since $\tau(Q_i), i=1, \ldots, 4$ have integer entries, one can suppose $x_i \in \Zp$.
Hence, $\tau(\prod_{i = 1}^4 Q_i^{x_i}) = \sum_{i=1}^{4} x_i \tau(Q_i) = \bzer$, which yields $I = \prod_{i = 1}^4 Q_i^{x_i} \in \langle \mG \rangle$.
\end{proof}

\subsection{$\mC$ has dimension 1}\label{subsection:case1}
Next, we consider the case where $\dim \mC = 1$.
The main idea of this case is that if the product of a string $B_1 \cdots B_m$ is in $\U_1$, then all $D_{\sigma}, E_{\sigma}, F_{\sigma}$ vanish, so $\tau(B(\sigma, t))$ is determined by some linear terms as well as by $G_{\sigma}$.
Recall that we write $A_i = UT(\alpha_i, \beta_i, \kappa_i; \delta_i, \epsilon_i, \phi_i)$, $i = 1, \ldots, k$.
Similar to notation \eqref{eq:defXi}, we define the following quantities for convenience:
\begin{equation}\label{eq:defGreeki}
\Delta_{i} \coloneqq \delta_i - \frac{1}{2}\alpha_i \beta_i, \quad
\mE_{i} \coloneqq \epsilon_i - \frac{1}{2}\beta_i \kappa_i, \quad
\Phi_{i} \coloneqq \alpha_i \beta_i \kappa_i - \frac{1}{2}(\alpha_i \epsilon_i + \delta_i \kappa_i) + \frac{1}{3}\phi_i.
\end{equation}
Since $\mC$ has dimension 1, there exist $\alpha, \beta, \kappa \in \Z$ such that $\varphi_0(A_i) = (\alpha, \beta, \kappa) \cdot \rho_i$ for $\rho_i \in \Z, i=1, \ldots, k$.

\begin{restatable}{prop}{caseone}\label{prop:case1}
    Suppose $\varphi_0(A_i) = (\alpha, \beta, \kappa) \cdot \rho_i$ for $\rho_i \in \Z, i=1, \ldots, k$.
    Let $\bl = (\ell_1, \ldots, \ell_k)$ be the Parikh vector of a string $B_1 \cdots B_m$, with the product $P = B_1 \cdots B_m$. Then 
    \begin{enumerate}[(i)]
        \item 
    $P \in \U_{10}$ if and only if $\sum_{i=1}^k \ell_i \rho_i = 0$ and $\sum_{i=1}^k \ell_i \Delta_i = 0$.
        \item 
        $P \in \U_2$ if and only if
    $\sum_{i=1}^k \ell_i \rho_i = 0, \sum_{i=1}^k \ell_i \Delta_i = 0$ and $\sum_{i=1}^k \ell_i \mE_i = 0$.
    \end{enumerate}
\end{restatable}
The immediate consequence of Proposition \ref{prop:case1} is that $\U_2$-Reachability and $\U_{10}$-Reachability are decidable using linear programming (LP).
For example, $\U_{10}$-Reachability has a positive answer if and only if the LP instance $\sum_{i=1}^k \ell_i \rho_i = 0$, $\sum_{i=1}^k \ell_i \Delta_i = 0$, $\ell_i \geq 0, i = 1, \ldots, k$, has a non-zero \emph{integer} solution $(\ell_1, \ldots, \ell_k)$.
However, because all the equations and inequalities in the LP instance are \emph{homogeneous}, the LP instance has a non-zero \emph{integer} solution if and only if it has a non-zero \emph{rational} solution.
Furthermore, the total bit length of $\rho_i, \Delta_i, \mE_i$ is linear with respect to the encoding size of $\mG$.
Therefore, the existence of a non-zero rational solution is decidable in polynomial time.
In particular, for $i = 1, \ldots, k$, one can decide whether this LP instance has a rational solution $(\ell_1, \ldots, \ell_k)$ with $\ell_i = 1$.
Then, the LP instance has a non-zero rational solution if and only if it has a rational solution with $\ell_i = 1$ for some $i$.
The decision procedure for $\U_2$-Reachability is similar.

Next, we consider the Identity Problem.
Define the set
\[
\Lambda \coloneqq \left\{ (\ell_1, \ldots, \ell_k) \in \Zp^k \middle| \sum_{i=1}^k \ell_i \rho_i = \sum_{i=1}^k \ell_i \Delta_i = \sum_{i=1}^k \ell_i \mE_i = 0 \right\}.
\]
By Proposition \ref{prop:case1}, the product of a string is in $\U_2$ if and only if its Parikh vector is in $\Lambda$.
It is easy to see that $\Lambda$ is additively closed, meaning $\Lambda + \Lambda \subseteq \Lambda$.
Define the \emph{support} of a Parikh vector $\bl = (\ell_1, \ldots, \ell_k)$ to be $\supp(\bl) = \{i \mid \ell_i \neq 0\}$,
and the support of the set $\Lambda$ to be
\[
\supp(\Lambda) \coloneqq \bigcup_{\bl \in \Lambda} \supp(\bl) = \{i \mid \exists (\ell_1, \ldots, \ell_k) \in \Lambda, \ell_i \neq 0\}.
\]
For $i = 1, \ldots, k$, we have $i \in \supp(\Lambda)$ if and only if the LP instance
$\sum_{i=1}^k \ell_i \rho_i = \sum_{i=1}^k \ell_i \Delta_i = \sum_{i=1}^k \ell_i \mE_i = 0$, $\ell_i > 0$ and $\ell_j \geq 0, j \neq i$ has an \emph{integer} solution.
Again, by homogeneity, this is decidable in polynomial time by deciding the existence of a \emph{rational} solution.
Hence, $\supp(\Lambda)$ is computable in polynomial time by deciding whether $i \in \supp(\Lambda)$ for all $i = 1, \ldots, k$.

If $\supp(\Lambda) \neq \{1, \ldots, k\}$, we can discard the elements $A_i \in \mG$ with $i \notin \supp(\Lambda)$, then $\card(\mG)$ decreases and we are done by the induction hypothesis.
Hence, we only need to consider the case where $\supp(\Lambda) = \{1, \ldots, k\}$.
The following proposition answers the Identity Problem in this case. 
Again, the homogeneity yields a polynomial time deciding procedure.
\begin{restatable}{prop}{caseoneid}\label{prop:case1id}
Suppose $\varphi_0(A_i) = (\alpha, \beta, \kappa) \cdot \rho_i$ for $\rho_i \in \Z, i=1, \ldots, k$, and $\supp(\Lambda) = \{1, \ldots, k\}$.
Define the values $\Gamma_i = \alpha \epsilon_i - \kappa \delta_i, i = 1, \ldots, k$.
Then
\begin{enumerate}[(i)]
    \item When $\rho_i \Gamma_j = \rho_j \Gamma_i$ for all $i, j \in \{1, \ldots, k\}$, the identity matrix is reachable if and only if the set $\{(\ell_1, \ldots, \ell_k) \in \Lambda \mid \sum_{i=1}^k \ell_i \Phi_i = 0\}$ is not equal to $\{\bzer\}$.
    \item When $\rho_i \Gamma_j \neq \rho_j \Gamma_i$ for some $i, j \in \{1, \ldots, k\}$, the identity matrix is reachable.
\end{enumerate}
\end{restatable}

\subsection{$\mC$ has dimension 0}\label{subsection:case0}
In this case, $\varphi_0(A_i) = \bzer$ for all $i$, so $\mG \subset \U_1$. 
Since $\U_1 \overset{\tau}{\cong} \Z^3$, the Identity Problem and $\U_2$, $\U_{10}$-Reachability are decidable using linear programming.
For example, deciding the Identity Problem amounts to deciding whether the LP instance $\sum_{i = 1}^k \ell_i \cdot \tau(A_i) = \bzer$, $\ell_i \geq 0, i = 1, \ldots, k$ has a non-zero integer solution.
As before, by the homogeneity of the LP instance, this is decidable in polynomial time by considering solutions in $\Q$.

\subsection{$\mC$ has dimension 2}

Suppose now that there exist $p, q, r \in \Z$, not all zero, such that $p \alpha_i + q \beta_i + r \kappa_i = 0, i=1, \ldots, k$.
Consider the following cases on the values of $p, q, r$.

\subsubsection{Case 1: there is at most one zero among $p, q, r$.}\label{subsubsection:case21}
The main difficulty of this case is as follows.
By Lemma \ref{lem:DER2}, $(D_{\sigma}, E_{\sigma})$ is constrained to the one dimensional subspace $\{(d,e) \mid pd - re = 0\} \subset \R^2$.
Therefore, in order to decide whether the vectors $\tau(B(\sigma, t))$ can generate the neutral element, one needs to take into account their linear terms, i.e.\ $(\sum_{i=1}^m D_i, \sum_{i=1}^m E_i)$ as well.
Define the additively closed set:
\[
L \coloneqq \left\{(\ell_1, \ldots, \ell_k) \in \Zp^k \middle| \sum_{i=1}^{k} \ell_i \varphi_0(A_i) = 0 \right\}.
\]
The product $P$ of a string $B_1 \cdots B_m$ is in $\U_1$ if and only if its Parikh vector is in $L$. 
\begin{restatable}{lem}{suppL}\label{lem:suppL}
When $\mC = \mC^{lin}$, we have $\supp(L) = \{1, \ldots, k\}$.
\end{restatable}

We continue to adopt the notations from \eqref{eq:defGreeki} for $\Delta_i, \mE_i$.
Consider the subset of $L$: 
    \[
    L_0 \coloneqq \left\{(\ell_1, \ldots, \ell_k) \in L \middle| p \sum_{i=1}^{k} \ell_i \Delta_i - r \sum_{i=1}^{k} \ell_i \mE_i = 0 \right\}.
    \]
$L_0$ can be described as the set of Parikh vectors whose corresponding strings have linear terms falling on the line $pd - re = 0$.
Again, $L_0$ is additively closed.
The main idea is that the quadratic term of $\varphi_1(B(\sigma, t))$ falls on the line $pd - re = 0$, therefore, if $P \in \U_2, \varphi_1(B(\sigma, t)) = 0$, then its linear term must also fall on the line $pd - re = 0$.
This leads to the following lemma.

\begin{restatable}{lem}{casetwoinLzero}\label{lem:case2inL0}
Suppose $\dim \mC = 2$. If the product $P$ of a string $B_1 \cdots B_m$ is in $\U_2$, then its Parikh vector $\bl$ is in $L_0$.
\end{restatable}

The following proposition gives a solution to the $\U_{10}$-Reachability problem.

\begin{restatable}{prop}{casetwoUonezero}\label{prop:case2U10}
Suppose $\dim \mC = 2$ and at most one of $p, q, r$ is zero.
\begin{enumerate}[(i)]
    \item When $r \neq 0$, $\U_{10}$ is reachable.
    \item When $r = 0, p \neq 0$, $\U_{10}$ is reachable if and only if $L_0$ is not equal to $\{\bzer\}$.
\end{enumerate}
\end{restatable}

In particular, whether $L_0$ equals $\{\bzer\}$ is decidable by linear programming, (again, by homogeneity, one can solve the linear programming instance in $\Q$).
Hence, $\U_{10}$-Reachability is decidable.
We then treat the Identity Problem and $\U_2$-Reachability.
Consider the support of $L_0$. 
As before, $\supp(L_0) = \{i \mid \exists (\ell_1, \ldots, \ell_k) \in L_0, \ell_i \neq 0\}$ is computable using linear programming.
By Lemma \ref{lem:case2inL0}, in order to reach $\U_2$ (or the identity matrix), we can only use matrices with index in $\supp(L_0)$.
By discarding matrices and using the induction hypothesis on $\card(\mG)$, we only need to consider the case where $\supp(L_0) = \{1, \ldots, k\}$.
The following proposition gives a positive answer to the Identity Problem and $\U_2$-Reachability in this case.

\begin{restatable}{prop}{casetwoonereachtwoD}\label{prop:case21reach2D}
Suppose $\dim \mC = 2$ and at most one of $p, q, r$ is zero. 
If $\supp(L_0) = \{1, \ldots, k\}$, then the identity matrix is reachable. (In particular, $\U_2$ is reachable.)
\end{restatable}
\begin{proof}[Sketch of proof]
Similarly to Proposition~\ref{prop:case3R3}, we construct four elements in $\U_1 \cap \langle \mG \rangle$ whose images under $\tau$ generate the two-dimensional linear subspace $\{(d,e,f) \in \R^3 \mid pd - re = 0\}$ as an $\Rp$-cone (see Figure \ref{fig:2dLattice} for an illustration). 
Consequently, the $\Zp$-cone that they generate is two-dimensional lattice in $\{(d,e,f) \in \Z^3 \mid pd - re = 0\}$, which contains the neutral element.
\end{proof}


\subsubsection{Case 2: $p = r = 0$.}\label{subsubsection:caseH5}
In this case, $\mG \subset \operatorname{H}_5$, so the Identity Problem is decidable by Lemma \ref{lem:Heis}.
$\U_2$ and $\U_{10}$-Reachability reduce to the Identity Problem in $\Z^4$ and $\Z^3$, respectively, which are decidable in polynomial time using linear programming.
Here, we claim an additional complexity result that strengthens Lemma \ref{lem:Heis}, which is crucial for a polynomial complexity algorithm for $\UT(4, \Z)$.

\begin{restatable}{prop}{Hncomp}\label{prop:Hncomp}
For a fixed $n$, the Identity Problem in $\operatorname{H}_{2n+1}$ is decidable in polynomial time.
\end{restatable}

\subsubsection{Case 3: $p = q = 0$, $r \neq 0$ or $r = q = 0$, $p \neq 0$.}\label{subsubsection:case23}
The main technique in this case is a reduction from the Identity Problem to $\U_2$-Reachability, from $\U_2$-Reachability to $\U_{10}$-Reachability, and from $\U_{10}$-Reachability to linear programming or to the Identity Problem in $\Heis$.
If $p = q = 0$, $r \neq 0$, then $\kappa_i = 0, i = 1, \ldots, k$.
If $r = q = 0$, $p \neq 0$, then $\alpha_i = 0, i = 1, \ldots, k$.
Define the following matrices in $\Heis$:
\[
H_i \coloneqq 
\begin{pmatrix}
        1 & \alpha_i & \delta_i \\
        0 & 1 & \beta_i \\
        0 & 0 & 1 \\
\end{pmatrix}, \quad i=1, \ldots, k.
\]
The following proposition along with Proposition~\ref{prop:Hncomp} provides a solution to $\U_{10}$-Reachability.
\begin{restatable}{prop}{casetworeducezero}\label{prop:case2reduce0}
\begin{enumerate}[(i)]
    \item When $\kappa_i = 0, i = 1, \ldots, k$, $\U_{10}$-Reachability for $A_1, \ldots, A_k$ is equivalent to the Identity Problem for $H_1, \ldots, H_k$.
    \item When $\alpha_i = 0, i = 1, \ldots, k$, $\U_{10}$ is reachable for $A_1, \ldots, A_k$ if and only if $\sum_{i = 1}^{k} \ell_{i}\delta_i = \sum_{i = 1}^{k} \ell_{i}\beta_i = \sum_{i = 1}^{k} \ell_{i}\kappa_i = 0$ has a non-zero integer solution $(\ell_{1}, \ldots, \ell_{k}) \in \Zp^k$.
\end{enumerate}
\end{restatable}

Next, consider the Identity Problem and $\U_2$-Reachability.
By symmetry, we can suppose $p = q = 0$, $r \neq 0$, so $\kappa_i = 0, i = 1, \ldots, k$.
Define 
\begin{equation}\label{eq:defAprime}
A_i' \coloneqq UT(\beta_i, \alpha_i, \epsilon_i; \delta_i, \phi_i, 0), i=1, \ldots, k,
\end{equation} 
the following proposition reduces the Identity Problem and $\U_2$-Reachability for $A_1, \ldots, A_k$ to reachability problems for $A_1', \ldots, A_k'$:
\begin{restatable}{prop}{casetworeduce}\label{prop:case2reduce}
Suppose $\kappa_i = 0, i = 1, \ldots, k$.
\begin{enumerate}[(i)]
    \item The Identity Problem for $A_1, \ldots, A_k$ is equivalent to $\U_2$-Reachability for $A_1', \ldots, A_k'$.
    \item $\U_2$-Reachability for $A_1, \ldots, A_k$ is equivalent to $\U_{10}$-Reachability for $A_1', \ldots, A_k'$.
\end{enumerate}
\end{restatable}

Together with the previous Subsections \ref{subsection:case3} - \ref{subsubsection:caseH5},
we have completely reduced the Identity Problem for $\mG$ to either the problem for a set of smaller cardinality, or to $\U_2$-reachability of another set.
We have also reduced $\U_2$-reachability for $\mG$ to either a problem for a set of smaller cardinality, or to $\U_{10}$-reachability of another set.
By Proposition~\ref{prop:case2reduce0} and the previous subsections, $\U_{10}$-reachability is decidable.
Hence, we have now exhausted all the possible cases for the dimension of $\mC$, and we conclude that the Identity Problem, $\U_2$-Reachability and $\U_{10}$-Reachability in $\UT(4, \Z)$ are decidable.

\section{Complexity analysis and concluding remarks}
In this paper, we have shown that the Identity Problem for $\UT(4, \Z)$ is decidable.
A brief analysis of our algorithm shows that it terminates in polynomial time.
In fact, we can first show that the algorithm for $\U_{10}$-Reachability terminates in polynomial time.
Starting with $k = \card(\mG)$ matrices, we need to solve at most $O(k)$ linear equations, $O(k)$ homogeneous linear programming instances and one Identity Problem in $\Heis$ before either $\card(\mG)$ decreases or a conclusion on $\U_{10}$-Reachability is reached.
All these problems have $O(k)$ inputs which are of polynomial size with respect to the coefficients of the matrices in $\mG$, and are known to have polynomial complexity. 
Furthermore, the number $\card(\mG)$ decreases at most $k$ times. Hence, the total complexity of our algorithm for $\U_{10}$-reachability is polynomial with respect to the input $\mG$.
Then, using the same method, we can show that the algorithm for $\U_{2}$-Reachability terminates in polynomial time:
since after polynomial time, either $\card(\mG)$ decreases, or the problem is reduced to $\U_{10}$-Reachability, or a conclusion on $\U_{2}$-Reachability is reached.
At last, we can show that the algorithm for the Identity Problem terminates in polynomial time:
after polynomial time, either $\card(\mG)$ decreases, or the problem is reduced to $\U_{2}$-Reachability or the Identity Problem in $\operatorname{H}_5$, or a conclusion on the Identity Problem is reached.
(In particular, the polynomial complexity of the Identity Problem in $\operatorname{H}_5$ is a new result of our paper, see Proposition~\ref{prop:Hncomp}.)

It is likely that our method can be adapted to study the Identity Problem for other metabelian matrix groups, for instance the direct product $\Heis^n$.
There is also evidence that the arguments in this paper can be strengthened to tackle the Identity Problem for $\UT(n, \Z)$ with $n \geq 5$, even though $\UT(5, \Z)$ ceases to be metabelian.
In fact, one can push the convex geometry arguments down the derived series of $\UT(n, \Z)$, even when the series has length greater than two.
Another natural follow-up question is the Membership Problem for $\UT(4, \Z)$.
An interesting idea would be to adapt the Register Automata method introduced in \cite{colcombet2019reachability} for passing from the Identity Problem to the Membership Problem.

\bibliography{UTFour}
\newpage
\appendix
\section{Algorithms for $\U_2$-reachability and $\U_{10}$-reachability}\label{app:alg}
\begin{algorithm}[H]
\caption{U2Reachability(): deciding $\U_2$-Reachability for a subset of $\UT(4, \Z)$.}
\label{alg:U2}
\begin{description}
\item[Input:] 
A set $\mathcal{G} = \{A_1, \ldots, A_k\}$ of matrices in $\UT(4, \Z)$.
\item[Output:] True or False.
\end{description}
\begin{enumerate}[Step 1:]
    \item Compute the cone $\mathcal{C}$ and its lineality space $\mathcal{C}^{lin}$.
For $i = 1, \ldots, k$, if some \\
$\varphi_0(A_i)$ is not in $\mathcal{C}^{lin}$, return U2Reachability($\mathcal{G} \setminus \{A_i\}$).
    \item
    \begin{enumerate}
        \item If $\dim(\mathcal{C}) = 3$, return True.
        \item If $\dim(\mathcal{C}) = 1$, return True if the condition in Proposition~\ref{prop:case1}(ii) can be satisfied, otherwise return False.
        \item If $\dim(\mathcal{C}) = 0$, return True if $\tau(A_i), i = 1, \ldots, m$ generate a semigroup intersecting $(0, 0,\Z)$, otherwise return False.
        \item If $\dim(\mathcal{C}) = 2$, compute a non-zero vector $(p,q,r) \in \Q^3$ orthogonal to $\mathcal{C}$.
    \begin{enumerate}
        \item If $p = 0$, but $q, r$ are not zero, or $r = 0$, but $q, p$ are not zero. Compute $L_0$, if $\supp(L_0)$ contains all matrices of $\mG$, return True, otherwise return U2Reachability($\{A_i \mid i \in \supp(L_0)\}$).
        \item If $p = r = 0$, problem reduces to linear programming.
        \item If $p = q = 0, r \neq 0$ or $r = q = 0, p \neq 0$, compute $A'_i$ as in \eqref{eq:defAprime}. Return U10Reachability($A_1', \ldots, A_k'$).
    \end{enumerate}
    \end{enumerate}
\end{enumerate}
\end{algorithm}
\begin{algorithm}[H]
\caption{U10Reachability(): deciding $\U_{10}$-Reachability for a subset of $\UT(4, \Z)$.}
\label{alg:U10}
\begin{description}
\item[Input:] 
A set $\mathcal{G} = \{A_1, \ldots, A_k\}$ of matrices in $\UT(4, \Z)$.
\item[Output:] True or False.
\end{description}
\begin{enumerate}[Step 1:]
    \item Compute the cone $\mathcal{C}$ and its lineality space $\mathcal{C}^{lin}$.
For $i = 1, \ldots, k$, if some \\
$\varphi_0(A_i)$ is not in $\mathcal{C}^{lin}$, return U10Reachability($\mathcal{G} \setminus \{A_i\}$).
    \item
    \begin{enumerate}
        \item If $\dim(\mathcal{C}) = 3$, return True.
        \item If $\dim(\mathcal{C}) = 1$, return True if the condition in Proposition~\ref{prop:case1}(i) can be satisfied, otherwise return False.
        \item If $\dim(\mathcal{C}) = 0$, return True if $\tau(A_i), i = 1, \ldots, m$ generate a semigroup intersecting $(0, \Z, \Z)$, otherwise return False.
        \item If $\dim(\mathcal{C}) = 2$, compute a non-zero vector $(p,q,r) \in \Q^3$ orthogonal to $\mathcal{C}$.
    \begin{enumerate}
        \item If $p = 0$, but $q, r$ are not zero. Return True.
        \item If $r = 0$, but $q, p$ are not zero. Compute $L_0$, return True if $L_0$ is non-empty, otherwise return False.
        \item If $p = r = 0$, Problem reduces to linear programming.
        \item If $p = q = 0, r \neq 0$, problem reduces to Identity Problem in $H_3$.
        \item If $r = q = 0, p \neq 0$, problem reduces to linear programming.
    \end{enumerate}
    \end{enumerate}
\end{enumerate}
\end{algorithm}

\section{Omitted proofs}\label{app:proofs}
\exactpr*
\begin{proof}
Denote by $\iota$ the neutral element of $\Sym_m$.
By symmetry, it suffices to prove the case where $\sigma$ is the identity permutation $\iota$, that is, with the values defined in Notation \ref{not:latin},
\begin{multline}
  B_1^t \cdots B_m^t = UT(t \sum_{i=1}^m a_i, t \sum_{i=1}^m b_i, t \sum_{i=1}^m c_i;\\ t^2 D_{\iota} + t \sum_{i=1}^m D_i, t^2 E_{\iota} + t \sum_{i=1}^m E_i, t^3 F_{\iota} + t^2 G_{\iota} + t \sum_{i=1}^m F_i).
\end{multline}

In fact,
\begin{align}\label{eq:Bit}
    B_i^t & = \exp(t \log (B_i)) \nonumber \\
    & = \exp\left(t 
    \left(
    \begin{pmatrix}
            0 & a_i & d_i & f_i \\
            0 & 0 & b_i & e_i \\
            0 & 0 & 0 & c_i \\
            0 & 0 & 0 & 0 \\
    \end{pmatrix}
    - \frac{1}{2} 
    \begin{pmatrix}
            0 & a_i & d_i & f_i \\
            0 & 0 & b_i & e_i \\
            0 & 0 & 0 & c_i \\
            0 & 0 & 0 & 0 \\
    \end{pmatrix}^2
    + \frac{1}{3}
    \begin{pmatrix}
            0 & a_i & d_i & f_i \\
            0 & 0 & b_i & e_i \\
            0 & 0 & 0 & c_i \\
            0 & 0 & 0 & 0 \\
    \end{pmatrix}^3
    \right)
    \right) \nonumber \\
    & = \exp\left(t 
    \begin{pmatrix}
            0 & a_i & d_i - \frac{1}{2}a_i b_i & f_i - \frac{1}{2}(a_i e_i + d_i c_i) + \frac{1}{3} a_i b_i c_i \\
            0 & 0 & b_i & e_i - \frac{1}{2}b_i c_i \\
            0 & 0 & 0 & c_i \\
            0 & 0 & 0 & 0 \\
    \end{pmatrix}
    \right) \nonumber \\
    & = I + t \begin{pmatrix}
            0 & a_i & D_i & F_i \\
            0 & 0 & b_i & E_i \\
            0 & 0 & 0 & c_i \\
            0 & 0 & 0 & 0 \\
    \end{pmatrix} + \frac{t^2}{2} \begin{pmatrix}
            0 & a_i & D_i & F_i \\
            0 & 0 & b_i & E_i \\
            0 & 0 & 0 & c_i \\
            0 & 0 & 0 & 0 \\
    \end{pmatrix}^2 + 
    \frac{t^3}{6} \begin{pmatrix}
            0 & a_i & D_i & F_i \\
            0 & 0 & b_i & E_i \\
            0 & 0 & 0 & c_i \\
            0 & 0 & 0 & 0 \\
    \end{pmatrix}^3 \nonumber \\
    & = \begin{pmatrix}
            1 & t a_i & \frac{t^2}{2} a_i b_i + t D_i & \frac{t^3}{6} a_i b_i c_i + \frac{t^2}{2} (a_i E_i + D_i c_i) + t F_i \\
            0 & 1 & t b_i & \frac{t^2}{2} b_i c_i + t E_i \\
            0 & 0 & 1 & t c_i \\
            0 & 0 & 0 & 1 \\
    \end{pmatrix} \nonumber \\
    & = UT(t a_i, t b_i, t c_i; \frac{t^2}{2} a_i b_i + t D_i, \frac{t^2}{2} b_i c_i + t E_i, \frac{t^3}{6} a_i b_i c_i + \frac{t^2}{2} (a_i e_i + d_i c_i - a_i b_i c_i) + t F_i)
\end{align}
Then, we use the following lemma.
\begin{lemma}\label{lem:prodexpr}
\begin{multline}\label{eq:prodexpr}
    \prod_{i = 1}^{m} UT(u_i, v_i, w_i; x_i, y_i, z_i)
    =
    UT\left(
            \displaystyle\sum_{i=1}^{m} u_i, \displaystyle\sum_{i=1}^{m} v_i,
            \displaystyle\sum_{i=1}^{m} w_i; \right.\\
\left.
\displaystyle\sum_{i < j} u_i v_j + \displaystyle\sum_{i=1}^{m} x_i,
            \displaystyle\sum_{i < j} v_i w_j + \displaystyle\sum_{i=1}^{m} y_i,
            \displaystyle\sum_{i < j < k} u_i v_j w_k + \displaystyle\sum_{i < j} (u_i y_j + x_i w_j) + \displaystyle\sum_{i=1}^{m} z_i
    \right)
\end{multline}
\end{lemma}
\begin{proof}
For $m = 1, 2, 3$, Equation \eqref{eq:prodexpr} can be verified directly.
For $m \geq 4$, we use induction.
Suppose that Equation \eqref{eq:prodexpr} is correct for $m - 1$, we prove it for $m$.

By the induction hypothesis, it suffices to prove
\begin{multline}
    UT\left(
            \displaystyle\sum_{i=1}^{m-1} u_i, \displaystyle\sum_{i=1}^{m-1} v_i,
            \displaystyle\sum_{i=1}^{m-1} w_i; 
\displaystyle\sum_{i < j \leq m-1} u_i v_j + \displaystyle\sum_{i=1}^{m-1} x_i,
            \displaystyle\sum_{i < j \leq m-1} v_i w_j + \displaystyle\sum_{i=1}^{m-1} y_i, \right. \\
            \left.
            \displaystyle\sum_{i < j < k \leq m-1} u_i v_j w_k + \displaystyle\sum_{i < j \leq m-1} (u_i y_j + x_i w_j) + \displaystyle\sum_{i=1}^{m-1} z_i
    \right)
    \cdot  UT(u_m, v_m, w_m; x_m, y_m, z_m) \\
    = UT\left(
            \displaystyle\sum_{i=1}^{m} u_i, \displaystyle\sum_{i=1}^{m} v_i,
            \displaystyle\sum_{i=1}^{m} w_i; 
\displaystyle\sum_{i < j} u_i v_j + \displaystyle\sum_{i=1}^{m} x_i,
            \displaystyle\sum_{i < j} v_i w_j + \displaystyle\sum_{i=1}^{m} y_i, \right. \\
            \left.
            \displaystyle\sum_{i < j < k} u_i v_j w_k + \displaystyle\sum_{i < j} (u_i y_j + x_i w_j) + \displaystyle\sum_{i=1}^{m} z_i
    \right).
\end{multline}
which can be verified directly.
\end{proof}
Apply Lemma \ref{lem:prodexpr} to the product of the expressions \eqref{eq:Bit} for $B_i^t, i = 1, \ldots, m$.

\begin{align*}
    & B_1^t \cdots B_m^t \\
    = & \prod_{i=1}^m UT(t a_i, t b_i, t c_i; \frac{t^2}{2} a_i b_i + t D_i, \frac{t^2}{2} b_i c_i + t E_i, \frac{t^3}{6} a_i b_i c_i + \frac{t^2}{2} (a_i e_i + d_i c_i - a_i b_i c_i) + t F_i) \\
    = & UT\left( t \sum_{i=1}^m a_i, t \sum_{i=1}^m b_i, t \sum_{i=1}^m c_i; \right. \\
    & \quad
    \sum_{i < j} t a_i \cdot t b_j + \sum_{i=1}^m \left( \frac{t^2}{2} a_i b_i + t D_i \right),
    \sum_{i < j} t b_i \cdot t c_j + \sum_{i=1}^m \left(\frac{t^2}{2} b_i c_i + t E_i\right), \\
    & \quad \sum_{i < j < k} t a_i \cdot t b_j \cdot t c_k + \sum_{i < j} \left(t a_i \cdot (\frac{t^2}{2} b_j c_j + t E_j) + (\frac{t^2}{2} a_i b_i + t D_i) \cdot t c_j \right) + \\
    & \quad \left. \sum_{i = 1}^m \left( \frac{t^3}{6} a_i b_i c_i + \frac{t^2}{2} (a_i e_i + d_i c_i - a_i b_i c_i) + t F_i) \right) \right) \\
    = & UT\left( t \sum_{i=1}^m a_i, t \sum_{i=1}^m b_i, t \sum_{i=1}^m c_i; \right. \\ 
    & \quad t^2 \left( \sum_{i < j} a_i b_j + \frac{1}{2}\sum_{i=1}^m a_i b_i \right) + t \sum_{i=1}^m D_i, t^2 \left( \sum_{i < j} b_i c_j + \frac{1}{2}\sum_{i=1}^m b_i c_i \right) + t \sum_{i=1}^m E_i, \\
    & \quad t^3 \left( \sum_{i < j < k} a_i b_j c_k + \frac{1}{2}\sum_{i < j} (a_i b_j c_j + a_i b_i c_j) + \frac{1}{6}\sum_{i = 1}^m a_i b_i c_i \right) \\
    & \quad + t^2 \left( \sum_{i < j} \left( a_i e_j - \frac{1}{2} a_i b_j c_j + d_i c_j - \frac{1}{2} a_i b_i c_j \right) + \frac{1}{2}\sum_{i=1}^m (a_i e_i + d_i c_i - a_i b_i c_i) \right) \\
    & \left. \quad + t \sum_{i=1}^m F_i \right) \\
    = & UT(t \sum_{i=1}^m a_i, t \sum_{i=1}^m b_i, t \sum_{i=1}^m c_i; t^2 D_{\iota} + t \sum_{i=1}^m D_i, t^2 E_{\iota} + t \sum_{i=1}^m E_i, t^3 F_{\iota} + t^2 G_{\iota} + t \sum_{i=1}^m F_i).
\end{align*}
\end{proof}

\DERtwo*
\begin{proof}
$(i) \Rightarrow (ii)$: For $1 \leq i < j \leq m$, let $\sigma_{i,j}$ denote a permutation such that $\sigma_{i,j}(i) = i, \sigma_{i,j}(i+1) = j$ and $s_{i, j} \in \Sym_m$ to be the permutation that swaps $i$ and $j$.
If for all $\sigma \in \Sym_m$, $p D_{\sigma} = r E_{\sigma}$, then in particular, $p D_{\sigma_{i,j}} = r E_{\sigma_{i,j}}$ and $p D_{s_{i, j} \circ \sigma_{i,j}} = r E_{s_{i, j} \circ \sigma_{i,j}}$.
Therefore, 
\[
p\left(D_{s_{i, j} \circ \sigma_{i,j}} - D_{\sigma_{i,j}}\right) = r \left(E_{s_{i, j} \circ \sigma_{i,j}} - E_{\sigma_{i,j}}\right).
\]
Writing out the exact expressions, this yields
\[
p\left(a_{j} b_{i} - a_{i} b_{j}\right) = r \left(b_{j} c_{i} - b_{i} c_{j}\right).
\]
which can be rewritten as
\begin{equation}\label{eq:lemR2}
b_i \left( pa_j + r c_j \right) = b_j \left( pa_i + r c_i \right).
\end{equation}
The same equation also holds for $i > j$ by symmetry.
We distinguish two cases:
\begin{enumerate}
    \item \textbf{If $b_i = 0$ for all $i$.}
    We have (ii) immediately.
    \item \textbf{If $b_i \neq 0$ for some $i$.} 
    Let $J$ be the set of indices $j$ such that $b_j = 0$ and $I$ be the set of indices $i$ such that $b_i \neq 0$. $I$ is not empty.
    Then for any $j \in J$ (if such $j$ exists), take some $i \in I$, since $b_j = 0, b_i \neq 0$, we have $p a_j + r c_j = 0$ by Equation \eqref{eq:lemR2}.
    Then, for any $i_1, i_2 \in I$, take Equation \eqref{eq:lemR2} with indices $i = i_1, j = i_2$, we have
    \[
    \frac{ pa_{i_1} + r c_{i_1} }{b_{i_1}} = \frac{ pa_{i_2} + r c_{i_2} }{b_{i_2}}.
    \]
    Note that the denominators do not vanish since $i_1, i_2 \in I$.
    Hence, there is a constant $q \in \R$ such that 
    \[
    \frac{ pa_{i} + r c_{i} }{b_{i}} = -q, \; i \in I.
    \]
    Therefore, $p a_i + q b_i + r c_i = 0$ for all $i \in I$.
    And trivially, $p a_j + q b_j + r c_j = 0$ for all $j \in J$.
\end{enumerate}

$(ii) \Rightarrow (i)$: If $b_i = 0$ for all $i$ then (i) is trivial.
Otherwise, suppose $p a_i + q b_i + r c_i = 0, i = 1, \ldots, m$.
By symmetry, it suffices to prove (i) in the case where $\sigma$ is the identity permutation.
Indeed,
\begin{align*}
    p \left(\sum_{i < j} a_i b_j + \frac{1}{2}\sum_{i=1}^{m}a_i b_i\right) & = \sum_{i < j} p a_{i} b_{j} + \frac{1}{2} \sum_{i=1}^{m} p a_{i} b_{i} \\
    & = - \sum_{i < j} q b_{i} b_{j} - \sum_{i < j} r c_{i} b_{j} - \frac{1}{2} \sum_{i=1}^{m} q b_i^2 - \frac{1}{2} \sum_{i=1}^{m} r b_i c_i \\
    & = -\frac{q}{2} (\sum_{i=1}^{m} b_i)^2 - r \left( \frac{1}{2} \sum_{i=1}^{m} b_i c_i + \sum_{i < j} c_{i} b_{j}\right) \\
    & = - r \left( \sum_{i=1}^m b_i \sum_{j=1}^m c_j - \frac{1}{2} \sum_{i=1}^{m} b_i c_i - \sum_{i < j} b_{i} c_{j} \right) \\
    & = r \left( \frac{1}{2} \sum_{i=1}^{m} b_i c_i + \sum_{i < j} b_{i} c_{j} \right).
\end{align*}
This shows (i) in the case where $\sigma$ is the identity permutation.
\end{proof}

\DEinv*
\begin{proof}
Take $\sigma'$ be such that $\sigma'(i) = \sigma(m+1-i), i = 1, \ldots, m$.
Then
\begin{align*}
    D_{\sigma'} & = \sum_{i < j} a_{\sigma'(i)} b_{\sigma'(j)} + \frac{1}{2}\sum_{i=1}^{m}a_{\sigma'(i)} b_{\sigma'(i)} \\
    & = \frac{1}{2} \sum_{i=1}^{m} a_i \sum_{i=1}^{m} b_i + \frac{1}{2}\sum_{i < j} \left(a_{\sigma'(i)} b_{\sigma'(j)} - a_{\sigma'(j)} b_{\sigma'(i)}\right) \\
    & = \frac{1}{2} \sum_{i < j} \left(a_{\sigma'(i)} b_{\sigma'(j)} - a_{\sigma'(j)} b_{\sigma'(i)}\right) \\
    & = \frac{1}{2} \sum_{i < j} \left(a_{\sigma(m+1-i)} b_{\sigma(m+1-j)} - a_{\sigma(m+1-j)} b_{\sigma(m+i-i)}\right) \\
    & = \frac{1}{2} \sum_{i > j} \left(a_{\sigma(i)} b_{\sigma(j)} - a_{\sigma(j)} b_{\sigma(i)}\right) \\
    & = - D_{\sigma}
\end{align*}
By analogy, $E_{\sigma'} = - E_{\sigma}$.
\end{proof}

\Fsumzero*
\begin{proof}
For any $\sigma \in \Sym_m$, let $\sigma'$ denote the permutation such that $\sigma' = \sigma(m+1-i), i = 1, \ldots, m$.
Note that $\sigma \mapsto \sigma'$ is a bijection between $\Sym_m$ and itself.
Consider a fixed pair of distinct indices $i, j$. 
Then, the proportion of $\sigma$ such that $\sigma(i) < \sigma(j)$ is exactly one half, the same as the proportion of $\sigma$ such that $\sigma(j) < \sigma(i)$.
Hence,
\begin{align*}
    & \sum_{\sigma \in \Sym_m} \frac{1}{2} \sum_{i < j}(a_{\sigma(i)} b_{\sigma(i)} c_{\sigma(j)} + a_{\sigma(i)} b_{\sigma(j)} c_{\sigma(j)}) \\
    = & \; \frac{m!}{4} \sum_{i \neq j}(a_{i} b_{i} c_{j} + a_{i} b_{j} c_{j}) \\
    = & \; \frac{m!}{4} \left( \sum_{i = 1}^{m} a_i b_i \sum_{j = 1}^{m} c_j - \sum_{i = 1}^{m} a_i b_i c_i + \sum_{i = 1}^{m} a_i \sum_{j = 1}^{m} b_j c_j - \sum_{i = 1}^{m} a_i b_i c_i \right) \\
    = & - \frac{m!}{2} \sum_{i = 1}^{m} a_i b_i c_i
\end{align*}

Consider a fixed triple $i, j, k$, two by two distinct. 
Then, the proportion of $\sigma$ such that $\sigma(i) < \sigma(j) < \sigma(k)$ is exactly one sixth, the same as the other orders.
Hence,
\begin{align*}
    & \sum_{\sigma \in \Sym_m} \sum_{i < j < k} a_{\sigma(i)} b_{\sigma(j)} c_{\sigma(k)} \\
    = & \; \frac{m!}{6} \sum_{i,j,k \text{ all distinct}} a_{i} b_{j} c_{k} \\
    = & \; \frac{m!}{6} \left( \sum_{i = 1}^{m} a_i \sum_{j = 1}^{m} b_j \sum_{k = 1}^{m} c_k - \sum_{i \neq j} a_i b_i c_j - \sum_{i \neq j} a_i b_j c_j - \sum_{i \neq j} a_i b_j c_i - \sum_{i = 1}^{m} a_i b_i c_i \right) \\
    = & \; \frac{m!}{6} \left( - \sum_{i=1}^{m} a_i b_i \sum_{j=1}^{m} c_j + \sum_{i = 1}^{m} a_i b_i c_i - \sum_{i=1}^{m} a_i \sum_{j=1}^{m} b_j c_j + \sum_{i = 1}^{m} a_i b_i c_i \right. \\
    & \left. - \sum_{i=1}^{m} a_i c_i \sum_{j=1}^{m} b_j + \sum_{i = 1}^{m} a_i b_i c_i - \sum_{i = 1}^{m} a_i b_i c_i \right) \\
    = & \; \frac{m!}{3} \sum_{i = 1}^{m} a_i b_i c_i
\end{align*}

Combining these, we get
\begin{align*}
    & \sum_{\sigma \in \Sym_m} F_{\sigma} \\
    = & \; \sum_{\sigma \in \Sym_m} \sum_{i < j < k} a_{\sigma(i)} b_{\sigma(j)} c_{\sigma(k)} + \sum_{\sigma \in \Sym_m} \frac{1}{2}\sum_{i < j}(a_{\sigma(i)} b_{\sigma(i)} c_{\sigma(j)} + a_{\sigma(i)} b_{\sigma(j)} c_{\sigma(j)}) \\
    & + \sum_{\sigma \in \Sym_m} \frac{1}{6}\sum_{i = 1}^m a_{i} b_{i} c_{i} \\
    = & \; \frac{m!}{3} \sum_{i = 1}^{m} a_i b_i c_i - \frac{m!}{2} \sum_{i = 1}^{m} a_i b_i c_i + \frac{m!}{6} \sum_{i = 1}^{m} a_i b_i c_i \\
    = & \; 0
\end{align*}
\end{proof}

\FRone*
\begin{proof}
Denote by $\mI$ the ideal of $\C[a_1, \ldots, a_4, b_1, \ldots, b_4, c_1, \ldots, c_4]$ generated by the 24 polynomials $F_{\sigma}, \sigma \in \Sym_4$, and the three polynomials $\sum_{i=1}^4 a_i, \sum_{i=1}^4 b_i, \sum_{i=1}^4 c_i$.

The variety $\operatorname{V}(\mI)$ can be decomposed into a union of irreducible varieties. To efficiently compute these irreducible varieties, we used a \emph{Primary Ideal Decomposition} algorithm \cite{cox2013ideals} in SageMath \cite{sagemath} to decompose $\operatorname{rad}(\mI)$, the radical of the ideal $\mI$.
The code for this computation is available at \url{https://doi.org/10.6084/m9.figshare.20121275.v1}.

Our algorithm returns that $\operatorname{V}(\mI)$ is the union of four irreducible varieties, defined by the ideals
\begin{align*}
    \mI_1 = & \langle a_1, a_2, a_3, a_4, \sum_{i=1}^4 b_i, \sum_{i=1}^4 c_i \rangle \\
    \mI_2 = & \langle b_1, b_2, b_3, b_4, \sum_{i=1}^4 a_i, \sum_{i=1}^4 c_i \rangle \\
    \mI_3 = & \langle c_1, c_2, c_3, c_4, \sum_{i=1}^4 a_i, \sum_{i=1}^4 b_i \rangle \\
    \mI_4 = & \langle \sum_{i=1}^4 a_i, \sum_{i=1}^4 b_i, \sum_{i=1}^4 c_i,
    b_4 c_3 - b_3 c_4, a_4 c_3 - a_3 c_4, b_4 c_2 - b_2 c_4, \\
    & b_3 c_2 - b_2 c_3, a_4 c_2 - a_2 c_4, a_3 c_2 - a_2 c_3,
    a_4 b_3 - a_3 b_4, a_4 b_2 - a_2 b_4, a_3 b_2 - a_2 b_3 \rangle
\end{align*}
Each of these ideals corresponds to one of the conditions in the statement.
\end{proof}

\caseone*
\begin{proof}
Let $B_1 \cdots B_m$ be a string such that $B_i \in \mG, i = 1, \ldots, m$. Denote its product $P = B_1 \cdots B_m$.
Let $\bl = (\ell_1, \ldots, \ell_k)$ denote its Parikh vector.
Denote $B_i = UT(a_i, b_i, c_i; d_i, e_i, f_i), i = 1, \ldots, m$, with $a_i = a r_i, b_i = b r_i, c_i = c r_i$ for some $r_i$.

It is clear that $P \in \U_1$ if and only if $\sum_{i=1}^m r_i = 0$,
which is equivalent to $\sum_{i=1}^k \ell_i \rho_i = 0$ by regrouping indices.

Denote by $\iota$ the neutral element of $\Sym_m$.
We use Proposition \ref{prop:exactexpr} with $t = 1, \sigma = \iota$.
We show that, if $\sum_{i=1}^m a_i = \sum_{i=1}^m b_i = \sum_{i=1}^m c_i = 0$, then $D_{\iota} = E_{\iota} = F_{\iota} = 0$.

Since $\sum_{i=1}^m a_i = \sum_{i=1}^m b_i = 0$, we have
\begin{align*}
    D_{\iota} & = \sum_{i < j} a_i b_j + \frac{1}{2}\sum_{i=1}^{m}a_i b_i \\
    & = \frac{1}{2} \sum_{i < j} (a_i b_j - a_j b_i) \\
    & = \frac{1}{2} \sum_{i < j} (a b r_i r_j - a b r_j r_i) \\
    & = 0.
\end{align*}

Similarly, we have $E_{\iota} = 0$.
By case (iv) of Lemma \ref{lem:FR1}, $F_{\iota} = 0$.

Therefore, by Proposition \ref{prop:exactexpr}, $B_1 \cdots B_m \in \U_{10}$ if and only if 
    \begin{equation*}
        \begin{cases}
            \sum_{i=1}^m a_i = \sum_{i=1}^m b_i = \sum_{i=1}^m c_i = 0 \\
            \sum_{i=1}^m D_i = 0.
        \end{cases}
    \end{equation*}
By regrouping the indices according to the Parikh vector, the above is equivalent to
    \begin{equation*}
        \begin{cases}
            \sum_{i=1}^k \ell_i \rho_i = 0 \\
            \sum_{i=1}^k \ell_i \Delta_i = 0.
        \end{cases}
    \end{equation*}
This proves (i).
Similarly, $B_1 \cdots B_m \in \U_2$ if and only if 
    \begin{equation*}
        \begin{cases}
            \sum_{i=1}^m a_i = \sum_{i=1}^m b_i = \sum_{i=1}^m c_i = 0 \\
            \sum_{i=1}^m D_i = 0 \\
            \sum_{i=1}^m E_i = 0.
        \end{cases}
    \end{equation*}
By regrouping the indices according to the Parikh vector, the above is equivalent to
    \begin{equation*}
        \begin{cases}
            \sum_{i=1}^k \ell_i \rho_i = 0 \\
            \sum_{i=1}^k \ell_i \Delta_i = 0 \\
            \sum_{i=1}^k \ell_i \mE_i = 0.
        \end{cases}
    \end{equation*}
This proves (ii).
\end{proof}

\caseoneid*
\begin{proof}
Consider any string $B_1 \cdots B_m \in \U_2$ with Parikh vector $\bl = (\ell_1, \ldots, \ell_k) \in \Lambda$.
Denote $B_i = UT(a_i = a r_i, b_i = b r_i, c_i = c r_i; d_i, e_i, f_i)$ for some $r_i, i = 1, \ldots, k$.
Since $B_1 \cdots B_m \in \U_2$, we have $\sum_{i=1}^m r_i = 0$.
Define $G_i = a e_i - c d_i, i = 1, \ldots, k$.
In the proof of Proposition \ref{prop:case1}, we have shown that $D_{\iota} = E_{\iota} = F_{\iota} = 0$, where $\iota$ is the neutral element of $\Sym_m$.
By symmetry, for all $\sigma \in \Sym_m$, $D_{\sigma} = E_{\sigma} = F_{\sigma} = 0$.
We show that
\begin{equation}\label{eq:Gsigma}
G_{\sigma} = \frac{1}{2} \sum_{i < j} (r_{\sigma(i)} G_{\sigma(j)} - r_{\sigma(j)} G_{\sigma(i)}).
\end{equation}
By symmetry, it suffices to show \eqref{eq:Gsigma} for $\sigma = \iota$. Indeed,
\begin{align*}
    G_{\iota} & = \sum_{i < j} (a_i e_j + d_i c_j - \frac{1}{2}a_i b_i c_j - \frac{1}{2}  a_i b_j c_j) + \frac{1}{2} \sum_{i=1}^m (a_i e_i + d_i c_i - a_i b_i c_i) \\
    & = \left( \sum_{i < j} a_i e_j + \frac{1}{2} \sum_{i=1}^m a_i e_i \right) + \left( \sum_{i < j} d_i c_j + \frac{1}{2} \sum_{i=1}^m d_i c_i \right) \\
    & \quad - \frac{1}{2} \left( \sum_{i < j} (a_i b_i c_j + a_i b_j c_j) + \sum_{i=1}^m a_i b_i c_i \right) \\
    & = \frac{1}{2} \left( \sum_{i=1}^m a_i \sum_{j=1}^m e_j + \sum_{i < j} (a r_i e_j - a r_j e_i) \right) + \frac{1}{2} \left( \sum_{i=1}^m d_i \sum_{j=1}^m c_j + \sum_{i < j} (d_i c r_j - d_j c r_i) \right) \\
    & \quad - \frac{abc}{2}\left( \sum_{i \neq j} r_i^2 r_j + \sum_{i=1}^m r_i^3 \right) \\
    & = \frac{1}{2} \sum_{i < j} \left(r_i (a e_j - c d_j) - r_j(a e_i - c d_i)\right) - \frac{abc}{2} \sum_{i=1}^m r_i \sum_{j=1}^m r_j^2 \\
    & = \frac{1}{2} \sum_{i < j} \left(r_i G_j - r_j G_i\right)
\end{align*}
(i) If $\rho_i \Gamma_j = \rho_j \Gamma_i$ for all $i, j \in \{1, \ldots, k\}$, then $G_{\sigma}$ vanishes for all $\sigma$.
Thus 
\[
P = B_1 \cdots B_m = UT(a \sum_{i=1}^k \ell_i \rho_i, b \sum_{i=1}^k \ell_i \rho_i, c \sum_{i=1}^k \ell_i \rho_i, \sum_{i=1}^k \ell_i \Delta_i, \sum_{i=1}^k \ell_i \mE_i, \sum_{i=1}^k \ell_i \Phi_i),
\]
and $P = I$ if and only if $\bl$ is in the set $\{(\ell_1, \ldots, \ell_k) \in \Lambda \mid \sum_{i=1}^k \ell_i \Phi_i = 0\}$.
This proves (i).

(ii) By the additivity of $\Lambda$, one can find an Parikh vector $\bl = (\ell_1, \ldots, \ell_k) \in \Lambda$ whose support is equal to $\supp(\Lambda) = \{1, \ldots, k\}$.
Let $P = B_1 \cdots B_m \in \U_2$ be a string with Parikh vector $\bl$.
Since $\rho_i \Gamma_j \neq \rho_j \Gamma_i$ for some $i, j \in \supp(\bl)$, we claim that there exists some non-zero $G_{\sigma}$.
Indeed, let $r_u G_v \neq r_v G_u$ for some $u, v \in \{1, \ldots, m\}$, let $\sigma_{uv} \in \Sym_m$ be a permutation such that $\sigma_{12}(1) = u, \sigma_{12}(2) = v$, and let $s_{12}$ be the permutation that swaps $1$ and $2$.
Then,
\[
G_{\sigma_{uv} \circ {s_{12}}} - G_{\sigma_{uv}} = r_v G_u - r_u G_v \neq 0.
\]
Hence at least one of $G_{s_{12} \circ \sigma_{uv}}$ and $G_{\sigma_{uv}}$ is non-zero.

Next, since $\sum_{\sigma \in \Sym_m} G_{\sigma} = 0$, one can find permutations $\sigma_+, \sigma_- \in \Sym_m$ such that $G_{\sigma_+} > 0, G_{\sigma_-} < 0$.
By additivity of $\Lambda$, $B(\sigma, t) \in \U_2$ for all $\sigma \in \Sym_m, t \in \Zp$.
Proposition \ref{prop:exactexpr} shows that, when $t \rightarrow \infty$,
\[
\tau_f(B(\sigma_{\pm}, t)) = t^2 G_{\sigma_{\pm}} + O(t).
\]
Therefore there exists a large enough $t$ such that
\[
\tau_f(B(\sigma_+, t)) > 0, \tau_f(B(\sigma_-, t)) < 0.
\]
Hence, there are $x_+, x_- \in \Zp$ such that $x_+ \tau_f(B(\sigma_+, t)) + x_- \tau_f(B(\sigma_-, t)) = 0$.
Consequently, $I = B(\sigma_+, t)^{x_+} B(\sigma_-, t)^{x_-} \in \langle \mG \rangle$.
This proves (ii).
\end{proof}

\suppL*
\begin{proof}
Take any index $i \in \{1, \ldots, k\}$,
we show that $i \in \supp(L)$.
Since 
\[
\langle \varphi_0(A_1), \ldots, \varphi_0(A_k) \rangle_{\Rp} = \mC = \mC^{lin}
\]
is a linear space, it contains $- \varphi_0(A_i)$.
Therefore, there exists $x_j \in \Rp, j = 1, \ldots, k$ such that $\sum_{j = 1}^{k} x_j \varphi_0(A_j) = - \varphi_0(A_i)$.
As all the entries are integers, we can suppose all $x_j$ lie in $\Qp$.
Let $d$ be a common denominator of $x_j, j = 1, \ldots, k$, then we have
\[
d(x_i + 1)\varphi_0(A_i) + \sum_{j \neq i} d x_j \varphi_0(A_j) = 0.
\]
Hence $(d x_1, \ldots, d(x_i + 1), \ldots, d x_k) \in L$. 
We conclude that $i \in \supp(L)$ since $d(x_i + 1) > 0$.
\end{proof}

\casetwoinLzero*
\begin{proof}
If $P \in \U_2$, then $\tau_d(P) = \tau_e(P) = 0$.
By Lemma \ref{lem:DER2}, $p D_{\iota} = r E_{\iota}$ (recall that $\iota \in \Sym_m$ denotes the identity permutation).
Therefore
\[
p \sum_{i=1}^{k} \ell_i \Delta_i - r \sum_{i=1}^{k} \ell_i \mE_i = \left(p D_{\iota} - r E_{\iota}\right) + p \sum_{i=1}^m D_i - r \sum_{i=1}^m E_i = p \tau_d(P) - r \tau_e(P) = 0,
\]
so $\bl \in L_0$.
\end{proof}

\casetwoUonezero*
\begin{proof}
(i). 
We prove this proposition by constructing two elements in $\U_1 \cap \langle \mG \rangle$ whose images under $\tau_d$ are positive and negative, respectively.

Since $L$ is additively closed, we can find an Parikh vector $\bl = (\ell_1, \ldots, \ell_k) \in L$ whose support is equal to $\supp(L) = \{1, \ldots, k\}$.
Consider a string $B_1 B_2 \cdots B_m \in \U_1$ whose Parikh vector is $\bl$.
Since $\dim\mC = 2$, and $\supp(\bl) = \{1, \ldots, k\}$, there exist $u, v \in \{1, \ldots, m\}$ such that $a_{u} b_{v} \neq a_{v} b_{u}$: otherwise for some $p', q' \in \R$, not both zero, we have $\varphi_0(A_i) \cdot (p',q',0)^{\top} = 0, i = 1, \ldots, k$, contradicting $r \neq 0$.
    
    Let $\sigma_{uv} \in \Sym_m$ be a permutation with $\sigma_{uv}(1) = u, \sigma_{uv}(2) = v$, and let $s_{12} \in \Sym_m$ be the permutation that swaps 1 and 2.
    Then
    \[
    D_{\sigma_{uv} \circ s_{12}} - D_{\sigma_{uv}} = a_{v} b_{u} - a_{u} b_{v} \neq 0.
    \]
    Hence, there exists a permutation $\sigma \in \Sym_m$ such that $D_{\sigma} \neq 0$.
    Consequently, by Lemma \ref{lem:DEinv},
    one can find $\sigma_+, \sigma_-$ such that $D_{\sigma_+} > 0, D_{\sigma_-} < 0$.
    Then, by Proposition \ref{prop:exactexpr}, as $t \rightarrow +\infty$,
    \[
    \tau_d(B(\sigma_{+}, t)) = t^2 D_{\sigma_{+}} + O(t), \quad \tau_d(B(\sigma_{-}, t)) = t^2 D_{\sigma_{-}} + O(t).
    \]
    Thus, one can find a large enough $t$ such that 
    \[
    \tau_d(B(\sigma_+, t)) > 0, \tau_d(B(\sigma_-, t)) < 0.
    \]
    Since $\tau_d(B(\sigma_{\pm}, t))$ are integers, there exist $x_+, x_- \in \Zp$ such that $x_+ \tau_d(B(\sigma_+, t)) + x_- \tau_d(B(\sigma_-, t)) = 0$.
    Consequently, $B(\sigma_+, t)^{x_+} B(\sigma_-, t)^{x_-} \in \U_{10}$.
    This proves (i).
    
(ii). Consider any string $B_1 B_2 \cdots B_m \in \U_1$.
Denote by $\bl = (\ell_1, \ldots, \ell_k) \in L$ its Parikh vector.
By Lemma \ref{lem:DER2} and $p \neq 0, r = 0$, we have $D_{\iota} = 0$. Then by Proposition \ref{prop:exactexpr} with $t = 1$ and $\sigma = \iota$, 
\[
    P \in \U_{10} \iff \sum_{i = 1}^{m} D_i = 0
    \iff \sum_{i = 1}^{m} \ell_i \Delta_i = 0
    \iff (\ell_1, \ldots, \ell_k) \in L_0.
\]
Hence, $\U_{10}$ is reachable if and only if $L_0$ is not $\{\bzer\}$.
This proves (ii).
\end{proof}

\casetwoonereachtwoD*
\begin{proof}
We prove this proposition by constructing four elements in $\U_1 \cap \langle \mG \rangle$ whose images under $\tau$ generate the two-dimensional linear subspace $\{(d,e,f) \in \R^3 \mid pd - re = 0\}$ as an $\Rp$-cone (see Figure \ref{fig:2dLattice} for an illustration). 
Consequently, they would generate a two-dimensional lattice in $\{(d,e,f) \in \Z^3 \mid pd - re = 0\}$ as a $\Zp$-cone (in other words, as an additive monoid). In particular, the identity matrix lies in $\U_1 \cap \langle \mG \rangle$ by the same argument as Corollary~\ref{cor:case3reachI}.
We proceed in two steps.
By symmetry, we can suppose $r \neq 0$.
\begin{enumerate}[1.]
    \item \textbf{Finding two vectors in $\varphi_1(\U_1 \cap \langle \mG \rangle)$ with directions $\pm (r, p)$.}
    
    Since $L_0$ is additively closed, we can find an Parikh vector $\bl = (\ell_1, \ldots, \ell_k) \in L_0$ whose support is equal to $\supp(L_0) = \{1, \ldots, k\}$.
    Consider a string $B_1 B_2 \cdots B_m \in \U_1$ whose Parikh vector is $\bl$.
    Denote $B_i = UT(a_i, b_i, c_i; d_i, e_i, f_i), i = 1, \ldots, m$.
    
    Since $\dim \mC = 2$ and $\supp(\bl) = \{1, \ldots, k\}$, for the same reason as in the proof of Proposition \ref{prop:case2U10}(i),
    one can find $\sigma_+, \sigma_-$ such that $D_{\sigma_+} > 0, D_{\sigma_-} < 0$.
    As $\bl \in L_0$ we have $
    p \sum_{i=1}^{k} \ell_i \Delta_i = r \sum_{i=1}^{k} \ell_i \mE_i$.
    Then by Proposition \ref{prop:exactexpr} and Lemma \ref{lem:DER2},
    \[
    \frac{\tau_e(B(\sigma_{\pm}, t))}{\tau_d(B(\sigma_{\pm}, t))} = \frac{t^2 E_{\sigma} + t \sum_{i = 1}^m E_i}{t^2 D_{\sigma} + t \sum_{i = 1}^m D_i} = \frac{t^2 E_{\sigma_{\pm}} + t \sum_{i=1}^{k} \ell_i \mE_i}{t^2 D_{\sigma_{\pm}} + t \sum_{i=1}^{k} \ell_i \Delta_i} = \frac{p}{r}, 
    \]
    \[
    \lim_{t \rightarrow +\infty} \tau_d(B(\sigma_{\pm}, t)) = \lim_{t \rightarrow +\infty} \left( t^2 D_{\sigma_{\pm}} + O(t) \right) = \pm \infty,
    \]
    Hence, we conclude that there exists a large enough $t_0 \in \Zp$, such that $\varphi_1(B(\sigma_+, t_0))$ and $\varphi_1(B(\sigma_-, t_0))$ generate the $\R$-linear space $\{(d,e) \in \R^2 \mid pd - re = 0\}$ as an $\Rp$-cone.

    \item \textbf{Finding two vectors in $\tau(\U_1 \cap \langle \mG \rangle)$ with directions arbitrarily close to $\pm (0, 0, 1)$.}
    
    Since $\dim \mC = 2$, let $B_1, B_2 \in \langle \mG \rangle$ be such that $\varphi_0(B_1)$ and $\varphi_0(B_2)$ are $\R$-linearly independent.
    Let $\bl_1, \bl_2$ be the Parikh vectors of $B_1, B_2$, respectively.
    As in the previous step, take $\bl \in L_0$ such that $\supp(\bl) = \{1, \ldots, k\}$.
    Define $m = \max\{\|\bl_1\|_{\infty}, \|\bl_2\|_{\infty}\}$,
    then $\bl_3 = m \bl - \bl_1 - \bl_2$ is in $L_0$ since all its entries must be non-negative. 
    Let $B_3$ be any string whose Parikh vector is $\bl_3$.
    We have $\sum_{i=1}^3 \varphi_0(B_i) = 0$
    because $\bl_1 + \bl_2 + \bl_3 = m\bl \in L_0$, as well as
    \begin{equation}\label{eq:case21pm}
        \langle \varphi_0(B_1), \varphi_0(B_2), \varphi_0(B_3)\rangle_{\Rp} = \mC = \{(a, b, c) \in \R^3 \mid pa + qb + rc = 0\}
    \end{equation}
    because $\varphi_0(B_1)$ and $\varphi_0(B_2)$ are $\R$-linearly independent and $\varphi_0(B_3) = - \varphi_0(B_1) - \varphi_0(B_2)$.
    
    Denote $B_i = UT(a_i, b_i, c_i; d_i, e_i, f_i), i = 1, 2, 3$.
    Take $(a_4, b_4, c_4) = (0, 0, 0)$ in Lemma \ref{lem:FR1}.
    The fact that at most one of $p, q, r$ is zero and Equation \eqref{eq:case21pm} yield that all four conditions in Lemma \ref{lem:FR1} are false. 
    Therefore there exist $\sigma \in \Sym_4$ such that $F_{\sigma} \neq 0$.
    Consequently, by Lemma \ref{lem:Fsum0}, there exist $\sigma'_+, \sigma'_- \in \Sym_4$ such that $F_{\sigma'_+} > 0, F_{\sigma'_-} < 0$.
    By ignoring $(a_4, b_4, c_4) = (0, 0, 0)$, this trivially implies that there exist $\sigma'_+, \sigma'_- \in \Sym_3$ such that $F_{\sigma'_+} > 0, F_{\sigma'_-} < 0$.
    
    Define the string
    \[
    B'(\sigma, t) = B_{\sigma(1)}^t B_{\sigma(2)}^t B_{\sigma(3)}^t \in \U_1.
    \]
    By Proposition \ref{prop:exactexpr} and Lemma \ref{lem:DER2}, 
    \[
    \frac{\tau_e(B'(\sigma'_{\pm}, t))}{\tau_d(B'(\sigma'_{\pm}, t))} = \frac{t^2 E_{\sigma} + t \sum_{i = 1}^m E_i}{t^2 D_{\sigma} + t \sum_{i = 1}^m D_i} = \frac{t^2 E_{\sigma} + t \sum_{i = 1}^k \ell_i \mE_i}{t^2 D_{\sigma} + t \sum_{i = 1}^k \ell_i \Delta_i} = \frac{p}{r}.
    \]
    Hence $\tau(B'(\sigma'_+, t)), \tau(B'(\sigma'_-, t))$ are in the linear subspace $\{(d,e,f) \in \R^3 \mid pd - re = 0\}$.
    Then by Proposition \ref{prop:exactexpr},
    \begin{align*}
    & \lim_{t \rightarrow +\infty} \frac{\tau_i(B'(\sigma'_{\pm}, t))}{\tau_f(B'(\sigma'_{\pm}, t))} = \lim_{t \rightarrow +\infty} \frac{O(t^2)}{t^3 F_{\sigma'_{\pm}} + O(t^2)} = 0, i \in \{d,e\}, \\
    & \lim_{t \rightarrow +\infty} \tau_f(B'(\sigma_{\pm}, t)) = \pm \infty.
        \end{align*}
    Note that the projections of $\tau(B(\sigma_+, t_0))$ and $\tau(B(\sigma_-, t_0))$ onto the $d, e$ coordinates generate the $\R$-linear space $\{(d,e) \in \R^2 \mid pd - re = 0\}$ as an $\Rp$-cone.
    By adding two vectors pointing close enough towards both directions of the $f$ axis, together the four vectors will generate the $\R$-linear space $\{(d,e,f) \in \R^3 \mid pd - re = 0\}$ as an $\Rp$-cone (see Figure \ref{fig:2dLattice} for an illustration).
    Hence, we conclude that there exists a large enough $t_1 \in \Zp$, such that $\tau(B(\sigma_+, t_0)), \tau(B(\sigma_-, t_0))$, $\tau(B'(\sigma'_+, t_1)), \tau(B'(\sigma'_-, t_1))$ generate $\{(d,e,f) \in \R^3 \mid pd - re = 0\}$ as an $\Rp$-cone.
    This shows that the identity matrix lies in $\U_1 \cap \langle \mG \rangle$.
\end{enumerate}
\end{proof}

\Hncomp*
\begin{proof}
For brevity we introduce the following notation.
\[
H(\boldsymbol{a}, \boldsymbol{b}, c) \coloneqq
\begin{pmatrix}
        1 & \boldsymbol{a} & c \\
        0 & I_{n} & \boldsymbol{b}^{\top} \\
        0 & 0 & 1 \\
\end{pmatrix}.
\]
Let $H_1, \ldots, H_k$ be a set of generators, $H_i = H(\boldsymbol{\alpha}_i, \boldsymbol{\beta}_i, \kappa_i), i = 1, \ldots, k$.

Consider the following $\Rp$-cone in $\R^{2n}$:
\[
\mC_H = \langle \{(\boldsymbol{\alpha}_i, \boldsymbol{\beta}_i) \mid i = 1, \ldots, k\} \rangle_{\Rp}.
\]
Denote $\mC_H^{lin}$ its lineality space.
For given $B_1, \ldots, B_m$, $B_i = H(\boldsymbol{a}_i, \boldsymbol{b}_i, c_i)$, denote $B(\sigma, t) = B_{\sigma(1)}^t \cdots B_{\sigma(m)}^t$.
One can show that
\[
B(\sigma, t) = H\left(t \sum_{i = 1}^m \boldsymbol{a}_i, t \sum_{i = 1}^m \boldsymbol{b}_i, t^2 C_{\sigma} + t \sum_{i=1}^m (c_i - \frac{1}{2} \boldsymbol{a}_i \cdot \boldsymbol{b}_i)\right),
\]
where 
\[
C_{\sigma} = \sum_{1 \leq i < j \leq m} \boldsymbol{a}_{\sigma(i)} \cdot \boldsymbol{b}_{\sigma(j)} + \frac{1}{2} \sum_{i = 1}^m \boldsymbol{a}_{i} \cdot \boldsymbol{b}_{i}.
\]

If some $H_i = H(\boldsymbol{\alpha}_i, \boldsymbol{\beta}_i, \kappa_i)$ satisfy $(\boldsymbol{\alpha}_i, \boldsymbol{\beta}_i) \notin \mC_H^{lin}$, then any string equal to the identity cannot contain $H_i$.
Therefore we can remove these $H_i$ and suppose $\mC_H = \mC_H^{lin}$.
Define
\[
L_H = \left\{ (\ell_1, \ldots, \ell_k) \in \Zp^k \middle| \sum_{i=1}^k \ell_i \boldsymbol{\alpha}_i = \sum_{i=1}^k \ell_i \boldsymbol{\beta}_i = \bzer \right\},
\]
which is an additively closed set.
We have $(\ell_1, \ldots, \ell_k) \in L_H$ if and only if $H_1^{\ell_1} \cdots H_k^{\ell_k} = H(\bzer, \bzer, c)$ for some $c$.
Similar to Lemma~\ref{lem:suppL}, we have $\supp(L_H) = \{1, \ldots, k\}$.

Consider the two following situations.
\begin{enumerate}
    \item If there exist $u, v$ such that $\boldsymbol{\alpha}_{u} \cdot \boldsymbol{\beta}_{v} \neq \boldsymbol{\alpha}_{v} \cdot \boldsymbol{\beta}_{u}$.
    We claim that $I \in \langle H_1, \ldots, H_k \rangle$.
    
    Since $L_H$ is additively closed, one can find a $B_1, \ldots, B_m$, $B_1 \cdots B_m = H(\bzer, \bzer, c)$ for some $c$, and such that every $H_i$ appears in $B_1, \ldots, B_m$ at least once.
    Write $B_i = H(\boldsymbol{a}_i, \boldsymbol{b}_i, c_i)$
    
    Let $\sigma_{uv} \in \Sym_m$ be a permutation with $\sigma_{uv}(1) = u, \sigma_{uv}(2) = v$, and let $s_{12}$ be the permutation that swaps 1 and 2. Then
    \[
    C_{\sigma_{uv} \circ s_{12}} - C_{\sigma_{uv}} = \boldsymbol{\alpha}_{u} \cdot \boldsymbol{\beta}_{v} - \boldsymbol{\alpha}_{v} \cdot \boldsymbol{\beta}_{u}.
    \]
    Hence, there exists a permutation $\sigma \in \Sym_m$ such that $C_{\sigma} \neq 0$. 
    Let $\sigma' \in \Sym_m$ be such that $\sigma'(i) = \sigma(m+1-i), i = 1, \ldots, m$. 
    We have that 
    \[
    C_{\sigma'} + C_{\sigma} = \sum_{i = 1}^m \boldsymbol{a}_i \sum_{j = 1}^m \boldsymbol{b}_j = 0,
    \]
    so $C_{\sigma'} = - C_{\sigma} \neq 0$.
    Therefore, when $t$ is large enough, the $c$-coordinates of $B(\sigma, t)$ and $B(\sigma', t)$ have different signs, thus the identity matrix can be generated by $B(\sigma, t)$ and $B(\sigma', t)$ as a semigroup.
    
    \item If $\boldsymbol{\alpha}_{u} \cdot \boldsymbol{\beta}_{v} = \boldsymbol{\alpha}_{v} \cdot \boldsymbol{\beta}_{u}$ for all $u, v \in \{1, \ldots, k\}$.
    
    Then $C_{\sigma} = 0$ whenever $B_1 \cdots B_m = H(\bzer, \bzer, c)$ for some $c$.
    Therefore $B_1 \cdots B_m = I$ if and only if
    $\sum_{i=1}^m \boldsymbol{a}_i = \sum_{i=1}^m \boldsymbol{b}_i = \sum_{i=1}^m (c_i - \frac{1}{2}\boldsymbol{a}_i \cdot \boldsymbol{b}_i) = 0$.
    The existence of such a string can be determined by linear programming.
\end{enumerate}
Notice that effectively computing $\mC_H^{lin}$ can be done in polynomial time. Also, all linear programming instances in the above procedure are of polynomial size with respect to the input $\{H_1, \ldots, H_k\}$.
Therefore, the overall complexity of the above procedure in polynomial in $k$.
\end{proof}

\casetworeducezero*
\begin{proof}
(i) $\langle \mG \rangle$ is contained in the following group
\begin{equation}\label{eq:defU02}
\U_{02} = \{UT(a, b, 0; d, e, f) \mid a, b, d, e, f \in \Z\}.
\end{equation}
Consider the group homomorphism
\begin{align*}
    \pi: \U_{02} & \rightarrow \Heis \\
    UT(a, b, 0; d, e, f) & \mapsto \begin{pmatrix}
            1 & a & d \\
            0 & 1 & b \\
            0 & 0 & 1 \\
    \end{pmatrix}
\end{align*}
$\pi$ is surjective and $\ker(\pi) = \U_{10}$.
Note that $\pi(A_i) = H_i$.
Therefore,
\begin{align*}
\U_{10} \cap \langle A_1, \ldots, A_k \rangle \neq \emptyset 
\iff & I \in \langle H_1, \ldots, H_k \rangle    
\end{align*}

(ii) $\langle \mG \rangle$ is contained in the following group
\[
\U_{00} = \{UT(0, b, c; d, e, f) \mid a, b, d, e, f \in \Z\}.
\]
Consider the group homomorphism
\begin{align*}
    \pi\colon \U_{00} & \rightarrow \Z^3 \\
    UT(0, b, c; d, e, f) & \mapsto (d, b, c)
\end{align*}
$\pi$ is surjective and $\ker(\pi) = \U_{10}$.
Note that $\pi(A_i) = (\delta_i, \beta_i, \kappa_i)$.
Therefore,
\begin{align*}
& \U_{10} \cap \langle A_1, \ldots, A_k \rangle \neq \emptyset \\
\iff & \bzer \in \langle (\delta_1, \beta_1, \kappa_1), \ldots, (\delta_k, \beta_k, \kappa_k) \rangle \\
\iff & \exists \bzer \neq (\ell_1, \ldots, \ell_k) \in \Zp^k, \sum_{i=1}^k \ell_i \delta_i = \sum_{i=1}^k \ell_i \beta_i = \sum_{i=1}^k \ell_i \kappa_i = 0
\end{align*}

\end{proof}

\casetworeduce*
\begin{proof}
(i) $\langle \mG \rangle$ is contained in $\U_{02}$ (defined in \eqref{eq:defU02}).
Define the group homomorphism
\begin{align*}
    \pi \colon \UT(4, \Z) & \rightarrow \U_{02} \\
    UT(b, a, e; d, f, *) & \mapsto UT(a, b, 0; d, e, f)
\end{align*}
$\pi'$ is surjective and $\ker(\pi) = \U_2$.
Moreover, $\pi(A_i') = A_i, i = 1, \ldots, k$. 
Thus,
\begin{align*}
\U_2 \cap \langle A_1', \ldots, A_k' \rangle \neq \emptyset 
\iff & I \in \langle A_1, \ldots, A_k \rangle
\end{align*}
This proves (i).

(ii) First of all, it is easy to show that 
\begin{align*}
    i \colon \U_{02} & \rightarrow \Heis \times \Z \\
    UT(a, b, 0; d, e, *) & \mapsto \left(H(b, a, d), e\right)
\end{align*}
is a surjective group homomorphism, where $H(b, a, d)$ denotes the matrix
\[
H = 
\begin{pmatrix}
        1 & b & d \\
        0 & 1 & a \\
        0 & 0 & 1 \\
\end{pmatrix}.
\]
Since $\ker(i) = \U_2$, $i$ induces a canonical isomorphism
\[
\bar{i} \colon \U_{02}/\U_2 \xrightarrow{\sim} \Heis \times \Z.
\]
Define the group homomorphism
\begin{align*}
    \pi' \colon \UT(4, \Z) & \rightarrow \Heis \times \Z \\
    UT(b, a, e; d, *, *) & \mapsto \left(H(b, a, d), e\right)
\end{align*}
$\pi'$ obviously surjective and $\ker(\pi') = \U_{10}$.

Composing $\pi'$ with the isomorphism $\bar{i}^{-1}$ then gives the surjective homomorphism
\[
\psi = \bar{i}^{-1} \circ \pi' \colon \UT(4, \Z) \rightarrow \U_{02}/\U_2.
\]
with $\ker(\psi) = \U_{10}$.
Moreover, we have $\psi(A_i') = \overline{A_i}, i = 1, \ldots, k$. 
Thus,
\begin{align*}
& \U_{10} \cap \langle A_1', \ldots, A_k' \rangle \neq \emptyset \\
\iff & \overline{I} \in \langle \overline{A_1}, \ldots, \overline{A_k} \rangle \; \text{ (in $\U_{02}/\U_2$)} \\
\iff & \U_2 \cap \langle A_1, \ldots, A_k \rangle \neq \emptyset
\end{align*}
This proves (ii).
\end{proof}

\end{document}